\documentclass[12pt]{article}

\usepackage{multicol}
\usepackage{multirow}
\usepackage{latexsym}
\usepackage{amsmath}
\usepackage{amsfonts}
\usepackage{amssymb}
\usepackage{psfrag}
\usepackage{graphicx}
\usepackage[usenames,dvipsnames]{xcolor}
\usepackage{natbib}

\usepackage{breakcites}
\usepackage{bbm}
\usepackage{epsfig,epstopdf}
\usepackage[colorlinks, linkcolor=black, citecolor=black]{hyperref}

\usepackage[utf8]{inputenc}
\usepackage[english]{babel}
\usepackage{amsthm}
\usepackage{caption}
\usepackage{setspace} 

\usepackage{sectsty}
\sectionfont{\fontsize{12}{15}\selectfont}
\subsectionfont{\fontsize{12}{15}\selectfont}

\newtheorem{theorem}{Theorem}
\newtheorem{lemma}{Lemma}

\def\no{\noindent}
\def \A{\boldsymbol{A}}
\def \S{\boldsymbol{S}}
\def \V{\boldsymbol{V}}

\begin{document}

{\centering {\large {\bf Semiparametric Inference of the Youden Index and the Optimal Cutoff Point under  Density Ratio Models}} \par}

\bigskip

\centerline{Meng Yuan, \ Pengfei Li \ and \ Changbao Wu\footnote{Meng Yuan is doctoral student, Pengfei Li is Professor and Changbao Wu is Professor, Department of Statistics and Actuarial Science, University of Waterloo, Waterloo ON N2L 3G1,Canada (E-mails: {\em m33yuan@uwaterloo.ca}, \ {\em pengfei.li@uwaterloo.ca} \ and \ {\em cbwu@uwaterloo.ca}).}}

\bigskip

\bigskip

\hrule

{\small
\begin{quotation}
\no
The Youden index is a popular summary statistic for receiver operating characteristic curve. 
It  gives the optimal cutoff point of a biomarker to distinguish the diseased and healthy individuals.
In this paper, we propose to model the distributions of a biomarker for individuals in the healthy and diseased groups
via a semiparametric density ratio model. 
Based on this model, we use the maximum empirical likelihood method to estimate the Youden index and the optimal cutoff point. 
We further establish the asymptotic normality of the proposed estimators and 
construct valid confidence intervals for the Youden index and the corresponding optimal cutoff point.
The proposed method automatically covers both cases when there is no lower limit of detection (LLOD) and when there is a fixed and finite LLOD for the biomarker.
Extensive simulation studies and a real data example are used to illustrate the effectiveness of the proposed method and its advantages  over the existing methods. 

\vspace{0.3cm}

\no
\end{quotation}
}

\hrule

\bigskip

\bigskip

\section{INTRODUCTION}
\label{intro}

Receiver operating characteristic (ROC) curve is a widely used statistical tool in medical research to 
evaluate the discriminatory effectiveness of a biomarker
for distinguishing diseased individuals from healthy ones.  
When the sampling distribution of the biomarker  is continuous, 
the ROC curve plots the proportion of true positive (sensitivity) versus proportion of false positive (one minus specificity)
across all possible choices of threshold values, called cutoff points, of the biomarker. 
We refer to \cite{zhou2002statistical}, \cite{pepe2003statistical}, \cite{krzanowski2009roc}, \cite{zou2011statistical, chen2016using}, and reference therein for comprehensive reviews and recent developments in ROC analysis.  

The Youden index, first proposed by \cite{youden1950index}, is one of popular summary statistics of the ROC curve. 
It is defined as the maximum of  the sum of sensitivity and specificity minus one 
when the relative seriousness of  false positive and false negative are treated equally. 
The Youden index ranges from 0 to 1 with 1 indicating  a complete separation of distributions of biomarkers in healthy  and diseased populations and 0 indicating a complete overlap. 
It has the advantage of providing a criterion to choose the ``optimal" cutoff point, which maximizes the sum of sensitivity and specificity minus one. 
See \cite{fluss2005estimation} for more discussions on the advantages of the Youden index. 

In medical researches, larger values of biomarkers are generally associated with diseases. Therefore, an individual is classified as diseased when the biomarker of the individual is greater than a given cutoff point. 
Let $F_0$ and $F_1$ denote the cumulative distribution functions (CDFs) of the healthy population and the diseased  population, respectively. 
Then the sensitivity and the specificity are respectively equal to $1 - F_1(x)$ and $F_0(x)$ for the given cutoff point $x$. 
Therefore,  the Youden index can be equivalently expressed as
\begin{equation}    \label{J}
    J= {\rm max}_x\{F_0(x) - F_1(x)\}= F_0(c) - F_1(c),
\end{equation}
where $c$ is the corresponding optimal cutoff point. 
In this paper, we aim to develop efficient inferential procedures for $J$ and $c$. 

In the literature, there are two types of methods, namely, the parametric method and the nonparametric method, for estimating the Youden index $J$ and the corresponding optimal cutoff point $c$.
For the parametric method, 
the original biomarkers or the biomarkers after the Box-Cox transformation \citep{box1964analysis} in the healthy and diseased groups
are assumed to come from the same parametric distribution family \citep{fluss2005estimation,bantis2019construction}.
The nonparametric method employs techniques such as the empirical CDF (ECDF) method  or the kernel method  to obtain the estimators of $F_0$ and $F_1$, 
which are then used to obtain the point estimators of $J$ and $c$. 
More details about the ECDF-based and kernel-based methods, and their modified versions
can be found in \cite{hsieh1996nonparametric}, \cite{zhou2012new}, and \cite{shan2015improved}. Recently, \cite{bantis2019construction} employed hazard constrained natural spline (HCNS) as an alternative nonparametric approach to estimate $J$ and $c$. 
The delta and bootstrap methods \citep{schisterman2007confidence,yin2014joint, bantis2019construction} and  the empirical likelihood (EL) methods \citep{wang2017smoothed} are used to construct confidence intervals (CIs) for $J$ and $c$.

In application, the measurement of a biomarker may have a fixed and finite lower limit of detection (LLOD). 
See, for example,  \cite{ruopp2008youden}, \cite{bantis2017estimation}, and the references therein.  
\cite{ruopp2008youden} adapted the parametric method, the ECDF method, and the ROC-generalized linear model (ROC-GLM) method \citep{pepe2000interpretation,alonzo2002distribution,pepe2003statistical} to obtain point estimates and construct CIs for $J$ and $c$ in this situation.

Generally speaking,  the parametric likelihood based estimators of $(J,c)$ are quite efficient, 
but may not be robust to model misspecifications \citep{fluss2005estimation}.
The nonparametric method is free from the model assumptions on $F_0$ and $F_1$, but
the resulting estimators of $(J,c)$, especially the estimator of $c$,  may be inefficient. 
When there is no LLOD, \cite{hsieh1996nonparametric} showed that 
the convergence rates of  the ECDF-based and the kernel-based estimators of $c$
are slower than $n^{-1/2}$, where $n$ is the total sample size.

In this paper, we develop a semiparametric method that enables efficient estimation of both $J$ and $c$
without making risky parametric assumptions on $F_0$ and $F_1$. 
In medical researches, the two populations under consideration usually share certain common characteristics \citep{zhuang2019semiparametric,qin2003using}. 
To incorporate the information from both samples, we suggest to use the density ratio model (DRM), proposed by \cite{anderson1979multivariate} and \cite{qin1997goodness}, to link
$F_0$ and $F_1$ as
\begin{equation}
    dF_1(x) = \exp\{\alpha + \boldsymbol{\beta}^T\mathbf{q}(x)\}dF_0(x) = \exp\{\boldsymbol{\theta}^T\mathbf{Q}(x)\}dF_0(x),
    \label{drm}
\end{equation}
where $dF_k(x)$ denotes the density of $F_k(x)$ for $k = 0,1$, the $\mathbf{q}(x)$ is a pre-specified, non-trivial function of dimension $p$ and $\boldsymbol{\theta}^T = (\alpha,\boldsymbol{\beta}^T)$ are unknown parameters. Note that $\mathbf{Q}(x)^T = (1, \mathbf{q}(x)^T)$. 
The unspecified baseline distribution $F_0$ makes DRM a semiparametric model. 
The DRM is quite flexible and includes many commonly used distribution families, such as normal, lognormal, and gamma distributions, as special cases. 

In the literature, DRMs have been used as a platform to  study inferential problems
for the ROC curve and  the area under the curve (AUC). 
Under the DRM, 
\cite{qin2003using} considered the estimation of the ROC curve and the AUC; 
\cite{zhang2006semiparametric} proposed a Wald-type statistic to test whether the accuracy of a diagnostic test is acceptable in terms of the AUC; 
\cite{wan2007smoothed} constructed a smoothed ROC curve estimator; 
and \cite{jiang2012inference} proposed two estimators for  the AUC with censored data. 
Inspired by \cite{qin2006empirical}, 
\cite{wang2014semiparametric} proposed an EL ratio based CI for the AUC under DRMs.
Later on, \cite{wan2008comparing} and \cite{ zhang2014semiparametric} considered the inference problems 
for the difference of AUCs for  two correlated ROC curves under a DRM.
Other applications of DRMs include 
multiple sample   hypothesis testing problems \citep{cai2017hypothesis,WANG2017, WANG2018}, quantile and quantile-function estimation \citep{chen2013quantile}, dominance index estimation \citep{zhuang2019semiparametric}.
More detailed reviews can be found in \cite{WangThesis}. 
In general, the inference procedures based on DRMs are more efficient than the fully nonparametric procedures.  To the best of our knowledge, the inference procedures for $(J,c)$ under a DRM have not been studied in the existing literature. 
This paper fills the void.

Our contributions can be summarized as follows. 
We construct the maximum EL estimators (MELEs)  of $(J,c)$ under a DRM based on the data with a LLOD, 
which automatically includes the case without a LLOD by setting the LLOD to be $-\infty$. 
We  establish the $\sqrt{n}$ convergence rates and the asymptotic normality of the proposed estimators of $(J,c)$
for data without a LLOD or with a fixed and finite LLOD. 
Our results show that when there is no LLOD the proposed estimator of $c$ has faster convergence rate than the existing nonparametric estimators,
and the proposed estimator of $J$ is asymptotically more efficient than the existing nonparametric estimators. 
When there is a fixed and finite LLOD, 
our proposed method is the first semiparametric or nonparametric method with rigorous theoretical justifications. 
Simulation experiments show that the proposed estimators are more efficient than or comparable to the nonparametric method
and are also comparable to parametric estimators under correctly specified distributions.  In addition, our proposed estimator for the optimal cutoff point $c$ has clear advantages over existing ones for all scenarios considered in the simulation. 

The rest of the paper is organized as follows. 
In Section \ref{sec3}, we propose the MELEs of $J$ and $c$ under a DRM
and study their asymptotic results. Confidence intervals of $J$ and $c$ are then constructed based on the asymptotic results. Simulation studies are presented in Section \ref{sec5} and 
a real data application is given in Section \ref{sec6}. We conclude with some discussion and additional remarks in Section \ref{sec7}. 
All  technical details are provided in the Appendix.

\section{MAIN RESULTS} \label{sec3}

\subsection{Point Estimation of $J$ and $c$ under the DRM}
\label{sec3.1}
Denote $\{x_{01},\ldots,x_{0n_0}\}$ and $\{x_{11},\ldots,x_{1n_1}\}$ as two independent random samples coming from the healthy and diseased populations, respectively. 
Let $f_0$ and $f_1$ be the probability density functions of $F_0$ and $F_1$, respectively. 
Following the definition of Youden index in \eqref{J}, the optimal cutoff point $c$ satisfies $f_0(c)=f_1(c)$, 
which together with \eqref{drm} implies that  
\begin{equation}\label{solution.c}
{\boldsymbol{\theta}}^T{\boldsymbol{Q}}(c)=0. 
\end{equation}
The above equation serves as the basis for estimating $c$. 

In the following, we focus on cases where the biomarker has a LLOD, denoted as $r$, and develop estimators for $(J,c)$. 
Analysis of data without a LLOD amounts to setting $r=-\infty$. 
Let $m_0$ and $m_1$ be the numbers of observations above the LLOD $r$ in the healthy and diseased groups, respectively. 
Let $\zeta_0=P(x_{01}\geq r)$ and $\zeta_1=P(x_{11}\geq r)$. 
Without loss of generality, we use $\{t_1,\ldots,t_m\}=\{x_{kj},k=0,1,~j=1,\ldots,m_k\}$ to denote the observations in the two samples which are above the LLOD, where $m=m_0+m_1$. 

We now discuss the maximum EL procedure for estimating the unknown parameters and functions. 
By the EL principle \citep{owen2001empirical} and under the DRM (\ref{drm}), 
the full likelihood can be written as 
\begin{eqnarray*}
\label{L_n}
L_n &=& \prod_{k=0}^1 \left[(1-\zeta_k)^{n_k - m_k} \prod_{j=1}^{m_k}dF_k(X_{kj})\right] \\
      &=& \prod_{k=0}^1 (1-\zeta_k)^{n_k - m_k}\prod_{i=1}^{m} p_i  \prod_{i=m_0+1}^m \exp\{\boldsymbol{\theta}^T \mathbf{Q}(t_i)\},
\end{eqnarray*} 
where $p_i=dF_0(t_i)$ for $i=1,\ldots,m$ and they satisfy the following constraints:
\begin{equation}
\label{constraint_censored}
p_i \geq 0,~ 0<\sum_{i = 1}^m p_i = \zeta_0 \leq1,~0<\sum_{i = 1}^m p_i \exp\{\boldsymbol{\theta}^T \mathbf{Q}(t_i)\} = \zeta_1 \leq1.
\end{equation}
The MELEs of $({\boldsymbol{\theta}},\zeta_0,\zeta_1,p_1,\ldots,p_m)$, denoted as $(\hat{\boldsymbol{\theta}},\hat\zeta_0,\hat\zeta_1,\hat p_{1},\ldots,\hat p_{m})$, 
are defined as the maximizers of $L_n$ subject to the constraints in Equation \eqref{constraint_censored}. 
It is shown by \cite{cai2018empirical} that 
\begin{equation}
\label{hat_zeta}
    \hat{\zeta}_k = m_k/n_k,~~ k = 0,1\,,
\end{equation}
and $\hat{\boldsymbol{\theta}}$ maximizes the following dual empirical log-likelihood function 
\begin{equation}
    \ell_{n}(\boldsymbol{\theta}) = \sum_{i = m_0+1}^m \{\boldsymbol{\theta}^T \mathbf{Q}(t_i)\} - \sum_{i = 1}^m\log\left[1+\rho \exp\{\boldsymbol{\theta}^T\mathbf{Q}(t_i)\}\right],
\end{equation}
where $\rho=n_1/n_0$. 
That is, 
 $\hat{\boldsymbol{\theta}}= \arg \max \limits_{\boldsymbol{\theta}}\ell_n(\boldsymbol{\theta})$. 
The MELEs of $p_i's$ are given by
\begin{equation}
    \label{pi_censored}
    \hat{p}_{i} = n_0^{-1}\left[1+\rho \exp\{\hat{\boldsymbol{\theta}}^T\mathbf{Q}(t_i)\}\right]^{-1},~~i = 1,\cdots,m.
\end{equation}
It follows that,  for any $x \geq r$, the MELEs of $F_0$ and $F_1$ are given by 
\begin{eqnarray*}
    \hat{F}_{0}(x) &=& (1 - \hat{\zeta}_0) + \frac{1}{n_0}\sum_{i=1}^{n}\frac{1}{1+\rho \exp\{\hat{\boldsymbol{\theta}}_r^T\mathbf{Q}(t_i)\}}I(r \leq t_i \leq x) \,,\\
    \hat{F}_{1}(x) &=& (1 - \hat{\zeta}_1) +\frac{1}{n_0}\sum_{i=1}^{n}\frac{ \exp\{\hat{\boldsymbol{\theta}}_r^T\mathbf{Q}(t_i)\} }{1+\rho \exp\{\hat{\boldsymbol{\theta}}_r^T\mathbf{Q}(t_i)\}}I(r \leq t_i \leq x) \,,
\end{eqnarray*}
where $I(\cdot)$ is the indicator function.

With the MELE $\hat{\boldsymbol{\theta}}$ and the equation \eqref{solution.c}, 
the MELE of the optimal cutoff point $c$, denoted as $\hat c$, can be obtained  as the solution to the equation
\begin{equation}
    \hat{\boldsymbol{\theta}}^T\mathbf{Q}(x) = 0.
    \label{drmpdf}
\end{equation} 
If multiple solutions exist for \eqref{drmpdf}  in $\left[\min t_i, \max  t_i\right]$, 
we choose the one that attains the maximum of $\hat{F}_{0}(x) - \hat{F}_{1}(x)$ as $\hat{c}$.
If a solution to \eqref{drmpdf} does not exist  in the range $\left[\min t_i, \max  t_i\right]$, 
we set $\hat c$ to be 
$$
\hat c=\arg\max_{x\in\{t_i:i=1,\ldots,m\}} \{ \hat{F}_{0}(x) - \hat{F}_{1}(x)\}. 
$$
The MELE $\hat{J}$ of $J$ is then given by 
$\hat{J} = \hat{F}_{0}(\hat{c}) - \hat{F}_{1}(\hat{c})$.

\subsection{Asymptotic Properties} 

In this section, we study the asymptotic properties of the MELEs $(\hat J,\hat c)$ described in Section \ref{sec3.1}.
We first introduce some further notation.  Let 
 $\boldsymbol{\theta}_0$ be the true value of $\boldsymbol{\theta}$
 and 
 $  \omega(x) = \exp \{\boldsymbol{\theta}_0^T\mathbf{Q}(x)\}$. 
For $t\geq r$, define
\begin{eqnarray*}
A_0(t)& =& \int_{r}^t \frac{\omega(x)}{1 + \rho\omega(x)} dF_0(x),\\
\A_1(t)&=&\int_{r}^t \frac{\omega(x)}{1 + \rho\omega(x)}\mathbf{q}(x) dF_0(x),\\
\A_2(t)&=& \int_{r}^t \frac{\omega(x)}{1 + \rho\omega(x)}\mathbf{q}(x)\mathbf{q}^T(x) dF_0(x).
\end{eqnarray*}
Further, let  
$ A_{0} = A_0(\infty)$, 
$\A_{1} = \A_1(\infty)$,
$ \A_2 = \A_2(\infty)$, 
and 
\begin{equation*}
    \A = \left ( \begin{array}{cc}
A_{0}& \A_{1}^T\\
\A_{1}&\A_{2}
\end{array} \right),~~\S=\frac{\rho}{1+\rho}\A,~~
\V = \S - \rho \left( \begin{array}{c}A_{0} \\ \A_{1} \end{array} \right) (A_{0}, \A_{1}^T).
\end{equation*} 
Define $\dot{\mathbf{q}}(x)=d{\mathbf{q}}(x)/dx$. 

\begin{theorem}
\label{thm1}
Let $(J_0,c_0)$ be the true value of $(J,c)$. 
Suppose the regularity conditions in the Appendix are satisfied and $c_0 > r$.  
As the total sample size $n=n_0+n_1$ goes to infinity, we have\\
(a) 
$\sqrt{n} (\hat{c} - c_0) \to {\rm N}(0,\sigma_{c}^2)$ in distribution, where 
    \begin{equation}
    \label{sigma_cr}
       \sigma_{c}^2 = \frac{\mathbf{Q}(c_0)^T \S^{-1}\V\S^{-1} \mathbf{Q}(c_0)}{\{\boldsymbol{\beta}_0^T\dot{\mathbf{q}}(c_0)\}^{2}}
    \end{equation}
  and $\boldsymbol{\beta}_0$ is the true value of $\boldsymbol{\beta}$;

\noindent
(b) 
$\sqrt{n} (\hat{J} - J_0) \to {\rm N} (0, \sigma_{J}^2)$ in distribution, 
   where
\begin{eqnarray}
\sigma_{J}^2 &=&  (\rho +1) \{F_0(c_0) - F_0^2(c_0)\}+ \frac{\rho+1}{\rho}\{F_1(c_0) - F_1^2(c_0)\} \nonumber \\
&& - \frac{(\rho +1)^3}{\rho}\left\{ A_{0}(c_0) - \left( \begin{array}{c}
A_{0}(c_0)\\ \A_{1}(c_0)\end{array} \right)^T \A^{-1} \left(\begin{array}{c}
A_{0}(c_0)\\ \A_{1}(c_0)\end{array}\right) \right\}.
\label{sigma_Jr}
\end{eqnarray}
\end{theorem}

We provide some comments on the results of Theorem 1.
Let $(\hat J_E, \hat c_E)$ and $(\hat J_K, \hat c_K)$ be  the ECDF-based and kernel-based estimators  
 of $(J,c)$, respectively. 
First, the estimators $\hat J$ and $\hat c$ 
both reach the convergence rates of the parametric likelihood based estimators. 
When there is no LLOD or $r=-\infty$, 
the convergence rate of $\hat c$ is faster than $\hat c_E$ and $\hat c_K$.
Second, when there is no LLOD or $r=-\infty$, 
\cite{hsieh1996nonparametric} showed that
\begin{equation*}
nE\{(\hat J_E-J_0)^2\} = \sigma^2_N+O(n^{-1/3}),~nE\{(\hat J_K-J_0)^2\}  = \sigma^2_N - \gamma n^{-v } \{1+o(1)\} 
\end{equation*}
for some $\gamma>0$,
where 
$$
\sigma^2_{N}=(\rho +1) \{F_0(c_0) - F_0^2(c_0)\}
+ \frac{\rho+1}{\rho}\{F_1(c_0) - F_1^2(c_0)\}.
$$
Here  the two bandwidths for the kernel method have the order $n^{-v}$ for some $0<v<1/3$.
According to Theorem 1 in \cite{qin1997goodness}, $\sigma^2_N-\sigma^2_{J}\geq0$.
Hence, when $n$ is large, the asymptotic mean square error of $\hat J$ is smaller than those of $\hat J_E$ and $\hat J_K$.

\subsection{Confidence Intervals on $J$ and $c$ under the DRM}

Replacing $(\boldsymbol{\theta}_0,J_0,c_0,F_0)$ in $\sigma^2_J$ and $\sigma^2_c$ by 
their respective estimators $(\hat {\boldsymbol{\theta}},\hat J,\hat c,\hat F_0)$, 
we obtain the estimators $(\hat \sigma^2_J, \hat \sigma^2_c)$  for $(\sigma^2_J, \sigma^2_c)$. 
It can be shown that $\hat \sigma^2_J$ and $\hat \sigma^2_c$ are both consistent. 

\begin{theorem}
\label{thm2}
Under the same conditions of Theorem 1, we have 
$$
\hat\sigma^2_J \to \sigma^2_J \;\;\; {\rm and} \;\;\; \hat\sigma^2_c \to \sigma^2_c
$$
in probability as $n\to\infty$. 
\end{theorem}

Because of the asymptotic normality of $\hat c$ presented in Theorem 1 and the consistency of $\hat \sigma^2_c$,
the quantity $\sqrt{n}(\hat c-c_0)/\hat\sigma_c$ is asymptotically pivotal, which leads to 
the following Wald-type CI for $c$ at level $1-a$: 
$$
\mathcal{I}_{c}=\left[\hat c-z_{1-a/2}\hat\sigma_c/\sqrt{n}, \hat c+z_{1-a/2}\hat\sigma_c/\sqrt{n}\right],
$$ 
where $z_{1-a/2}$ is the $100(1-a/2)$th quantile of the standard normal distribution. 

We can similarly construct a Wald-type CI for $J$. 
Our simulation experience indicates that a logit transformation on $\hat J$ leads to 
a CI for $J$ with better coverage accuracy, especially when $J_0$ is close to 0 or 1. 
More specifically,  using the results in Theorems 1 and 2, it can be shown that 
$$
\sqrt{n}\{\mbox{logit}(\hat J)-\mbox{logit}(J_0)\}\to N\left(0,\frac{\sigma_J^2}{J_0^2(1-J_0)^2}\right)
$$
in distribution as $n\to\infty$, where $\mbox{logit}(x)=\log\{x/(1-x)\}$ for $0<x<1$. 
Hence 
$\sqrt{n}\hat J(1-\hat J)\{\mbox{logit}(\hat J)-\mbox{logit}(J_0)\}/\hat\sigma_J$
is also asymptotically pivotal. 
This suggests the following CI for $J$: 
$$
\mathcal{I}_{J}=\left[
\mbox{expit}\left\{\mbox{logit}(\hat J)- \frac{z_{1-a/2}\hat\sigma_J}{\sqrt{n}\hat J(1-\hat J)}\right\},
\mbox{expit}\left\{\mbox{logit}(\hat J)+\frac{z_{1-a/2}\hat\sigma_J}{\sqrt{n}\hat J(1-\hat J)}\right\}\right],
$$ 
where $\mbox{expit}(x)=\exp(x)/\{1+\exp(x)\}$.

\section{SIMULATION STUDIES}
\label{sec5}

\subsection{Candidate methods}

In this section, we report results from simulation studies to compare the proposed point estimators  and CIs  of $(J,c)$ 
with  the following candidate methods.

\begin{itemize}
    \item[--] The Box-Cox method in \cite{bantis2019construction}, 
    where    the corresponding point estimators and CIs of $(J,c)$
    are denoted as $(\hat J_B,\hat c_B)$ and 
$(\mathcal{I}_{JB},\mathcal{I}_{cB})$, respectively.     
    \item[--] The ROC-GLM method in \cite{ruopp2008youden}, 
    where the corresponding point estimators and CIs of $(J,c)$
    are denoted as  $(\hat J_G,\hat c_G)$ and 
$(\mathcal{I}_{JG},\mathcal{I}_{cG})$, respectively.     
\item[--] The ECDF-based method, where 
the corresponding point estimators and CIs of $(J,c)$
    are denoted as $(\hat J_E,\hat c_E)$ and 
$(\mathcal{I}_{JE},\mathcal{I}_{cE})$, respectively.
 \item[--] The kernel-based method in \cite{fluss2005estimation}, 
 where the corresponding point estimators and CIs of $(J,c)$
    are denoted as $(\hat J_K,\hat c_K)$ and 
$(\mathcal{I}_{JK},\mathcal{I}_{cK})$, respectively. 
    \item[--] The HCNS method in \cite{bantis2019construction}, 
    where the corresponding point estimators and CIs of $(J,c)$
    are denoted as $(\hat J_H,\hat c_H)$ and 
$(\mathcal{I}_{JH},\mathcal{I}_{cH})$, respectively. 
     \end{itemize}
For all the above candidate methods, except for $\mathcal{I}_{JB}$, which  is obtained via the delta method, 
the CIs are constructed using the nonparametric bootstrap percentile method. 

When there is no LLOD, we compare our proposed method and all the candidate methods listed above. 
When there is a fixed  and finite LLOD,  to the best of our knowledge, 
the kernel-based method and the HCNS method have not been explored in the literature, 
and hence we do not include these two methods in our comparisons. 

\subsection{Simulation setup}
We conduct simulation studies under the following two distributional settings from \cite{fluss2005estimation}: 
\begin{itemize}
    \item [] (1) $f_0  \sim {\rm Gamma}(2,0.5)~{\rm and}~f_1  \sim {\rm Gamma}(2,\eta);$
    \item [] (2) $f_0  \sim {\rm LN}(2.5,0.09)~{\rm and} ~f_1  \sim {\rm LN}(\mu,0.25).$
\end{itemize}
Here ${\rm Gamma}(\kappa,\eta)$ denotes the  gamma distribution with shape parameter $\kappa$ and rate parameter $\eta$ and ${\rm LN}(\mu, \sigma^2)$ denotes the lognormal distribution with mean $\mu$ and variance $\sigma^2$, both with respect to the log scale.
The proposed estimators are calculated under the correctly specified ${\bf q}(x)$. 
For the gamma distributional setting, ${\bf q}(x)=x$, and for lognormal distribution setting,  ${\bf q}(x)=(\log x,\log ^2x)^T$. 

For each distributional setting, we choose four values of $\eta$ or $\mu$  such that the corresponding 
Youden indexes equal  0.2, 0.4, 0.6, and 0.8. 
The details are in Table 1 of the Supplementary Material. 
For the LLOD, we consider three values: $-\infty$, 15\% quantile of $F_0$, and 30\% quantile of $F_0$. 
The exact values of LLOD for  the latter two cases are in   Table 2 of the Supplementary Material.
Note that when the LLOD equals $-\infty$, there is no LLOD. 
For each simulation scenario, we consider five sample size combinations: $(n_0,n_1) = (50, 50)$,  $(100, 100)$, $(200, 200)$, $(150, 50)$, and $(50, 150)$, and results are based on 1,000 repeated simulation runs. 

The simulation results from different simulation scenarios demonstrate similar patterns.
To save space, we only report the simulation results under the gamma distributional setting with no LLOD and with the LLOD equal to the 15\% quantile of $F_0$. 
Other simulation results are provided in the Supplementary Material. 

\subsection{Comparison for point estimators} 

We first examine the point estimators of $(J,c)$. 
The performance of an point estimator is evaluated through the relative bias (RB) in percentage 
and mean squared error (MSE) computed as 
\begin{equation*}
    RB = \frac{1}{B}\sum_{b=1}^{B}\frac{a^{(b)} - a_0}{a_0} \times 100,~~~ MSE = \frac{1}{B}\sum_{b=1}^{B}(a^{(b)} - a_0)^2,
\end{equation*}
where $a_0$ is the true value of the interested quantity, $a^{(b)}$ is the estimator of the quantity computed from the $b$th simulation run, and $B=1,000$ is the number of simulation runs. 
The simulation results are presented in Tables~\ref{point.J.no}--\ref{point.c.lod}.

\begin{table}[!htt]
\centering
\small
\caption{RB (\%) and MSE ($\times 100$) for point estimators of  $J$ when there is no LLOD}
\label{point.J.no}
\begin{tabular}{cccccccccccc}
  \hline
 &$(n_0, n_1)$& \multicolumn{2}{c}{$(50, 50)$} & \multicolumn{2}{c}{$(100, 100)$} & \multicolumn{2}{c}{$(200, 200)$}& \multicolumn{2}{c}{$(50, 150)$} & \multicolumn{2}{c}{$(150, 50)$}\\ 
  \hline
  $J$& & RB & MSE & RB & MSE& RB & MSE& RB & MSE& RB & MSE\\ \hline
  \multirow{6}{*}{$0.2$}& $\hat{J}$  & 4.26 & 0.61 & 1.74 & 0.31 & 0.73 & 0.14 & 3.35 & 0.39 & 2.01 & 0.42 \\ 
  &$\hat{J}_{B}$  & 8.53 & 0.62 & 4.17 & 0.32 & 2.48 & 0.15 & 6.61 & 0.41 & 5.40 & 0.43 \\ 
  &$\hat{J}_{G}$ & 10.08 & 0.63 & 5.03 & 0.33 & 2.84 & 0.15 & 8.16 & 0.41 & 6.04 & 0.45 \\ 
  &$\hat{J}_{E}$ & 40.20 & 1.29 & 26.63 & 0.64 & 17.72 & 0.30 & 33.05 & 0.86 & 32.69 & 0.89 \\ 
  &$\hat{J}_{K}$ & 9.45 & 0.66 & 5.28 & 0.35 & 3.07 & 0.17 & 6.12 & 0.45 & 8.03 & 0.47 \\ 
  &$\hat{J}_{H}$ & 17.27 & 0.78 & 7.96 & 0.37 & 5.39 & 0.22 & 12.91 & 0.49 & 9.72 & 0.55 \\ \hline
  \multirow{6}{*}{$0.4$}& $\hat{J}$  & 2.43 & 0.57 & 0.95 & 0.29 & 0.41 & 0.13 & 1.62 & 0.35 & 1.18 & 0.40 \\ 
  &$\hat{J}_{B}$  & 4.47 & 0.58 & 2.32 & 0.29 & 1.58 & 0.14 & 3.37 & 0.36 & 2.81 & 0.40 \\ 
  &$\hat{J}_{G}$ & 3.38 & 0.56 & 1.63 & 0.29 & 1.07 & 0.14 & 2.40 & 0.34 & 2.37 & 0.42 \\ 
  &$\hat{J}_{E}$ & 16.42 & 1.06 & 10.62 & 0.53 & 6.78 & 0.24 & 13.36 & 0.68 & 12.57 & 0.74 \\ 
  &$\hat{J}_{K}$ & 2.70 & 0.57 & 1.08 & 0.31 & 0.49 & 0.15 & 1.57 & 0.38 & 1.86 & 0.41 \\ 
  &$\hat{J}_{H}$ & 4.68 & 0.68 & 1.33 & 0.33 & 0.96 & 0.18 & 3.15 & 0.40 & 1.71 & 0.52 \\ \hline
  \multirow{6}{*}{$0.6$}& $\hat{J}$  & 1.61 & 0.45 & 0.56 & 0.24 & 0.27 & 0.11 & 0.93 & 0.26 & 0.81 & 0.34 \\ 
  &$\hat{J}_{B}$  & 2.95 & 0.44 & 1.58 & 0.22 & 1.13 & 0.11 & 2.13 & 0.26 & 1.87 & 0.32 \\ 
  &$\hat{J}_{G}$ & 1.14 & 0.42 & 0.37 & 0.23 & 0.30 & 0.11 & 0.42 & 0.25 & 1.16 & 0.35 \\ 
  &$\hat{J}_{E}$ & 8.89 & 0.76 & 5.38 & 0.37 & 3.71 & 0.19 & 7.14 & 0.48 & 6.73 & 0.56 \\ 
  &$\hat{J}_{K}$ & 0.04 & 0.39 & -0.63 & 0.22 & -0.61 & 0.11 & -0.07 & 0.27 & -0.53 & 0.29 \\ 
  &$\hat{J}_{H}$ & 1.78 & 0.61 & 0.16 & 0.28 & 0.05 & 0.15 & 0.96 & 0.31 & 0.39 & 0.47 \\ \hline
  \multirow{6}{*}{$0.8$}& $\hat{J}$  & 1.06 & 0.26 & 0.38 & 0.14 & 0.16 & 0.06 & 0.55 & 0.13 & 0.63 & 0.21 \\ 
  &$\hat{J}_{B}$  & 1.56 & 0.22 & 0.81 & 0.11 & 0.60 & 0.05 & 1.07 & 0.12 & 0.96 & 0.17 \\ 
  &$\hat{J}_{G}$ & -0.26 & 0.26 & -0.36 & 0.14 & -0.32 & 0.07 & -0.63 & 0.15 & 0.25 & 0.22 \\ 
  &$\hat{J}_{E}$ & 4.51 & 0.39 & 2.86 & 0.20 & 1.94 & 0.1 & 3.73 & 0.24 & 3.53 & 0.30 \\ 
  &$\hat{J}_{K}$ & -2.38 & 0.24 & -2.40 & 0.15 & -2.00 & 0.08 & -1.74 & 0.14 & -2.88 & 0.21 \\ 
  &$\hat{J}_{H}$ & 1.55 & 0.40 & 0.71 & 0.19 & 0.75 & 0.11 & 1.01 & 0.19 & 0.82 & 0.32 \\\hline
\end{tabular}
\end{table}

\begin{table}[!htt]
\centering
\small
\caption{RB (\%) and MSE ($\times 100$) for point estimators of  $c$ when there is no LLOD}
\label{point.c.no}
\begin{tabular}{cccccccccccc}
  \hline
 &$(n_0, n_1)$& \multicolumn{2}{c}{$(50, 50)$} & \multicolumn{2}{c}{$(100, 100)$} & \multicolumn{2}{c}{$(200, 200)$}& \multicolumn{2}{c}{$(50, 150)$} & \multicolumn{2}{c}{$(150, 50)$}\\ 
  \hline
  $J$& & RB & MSE & RB & MSE& RB & MSE& RB & MSE& RB & MSE\\ \hline
  \multirow{6}{*}{$0.2$}& $\hat{c}$ & -0.39 & 12.22 & -0.15 & 6.09 & -0.03 & 2.78 & -0.40 & 8.56 & -0.02 & 7.68 \\ 
  &$\hat{c}_{B}$ & -3.03 & 81.16 & -1.97 & 41.56 & -2.33 & 21.17 & -1.83 & 56.07 & -2.85 & 56.79 \\ 
   &$\hat{c}_{G}$ & -6.33 & 108.4 & -2.55 & 56.42 & -2.20 & 29.56 & -4.64 & 70.96 & -1.79 & 73.30 \\ 
  &$\hat{c}_{E}$  & -0.92 & 249.25 & 0.49 & 170.90 & -0.76 & 116.8 & -0.78 & 201.93 & 1.89 & 223.79 \\ 
  &$\hat{c}_{K}$ & 16.99 & 405.52 & 10.39 & 164.68 & 5.99 & 91.84 & 13.63 & 287.86 & 13.00 & 254.68 \\ 
  &$\hat{c}_{H}$ & 1.75 & 237.90 & 3.16 & 171.47 & 2.71 & 110.91 & -1.09 & 179.73 & 7.6 & 243.96 \\ \hline
 \multirow{6}{*}{$0.4$}& $\hat{c}$ & -0.32 & 18.88 & -0.13 & 9.35 & 0.00 & 4.27 & -0.42 & 13.68 & 0.02 & 11.05 \\ 
 &$\hat{c}_{B}$ & -3.25 & 42.84 & -2.37 & 21.78 & -2.22 & 11.46 & -2.58 & 28.62 & -2.53 & 26.45 \\ 
  &$\hat{c}_{G}$ & -5.73 & 80.28 & -3.26 & 40.06 & -2.15 & 19.16 & -4.16 & 46.09 & -1.75 & 37.58 \\ 
  &$\hat{c}_{E}$  & -2.89 & 160.26 & -0.23 & 126.83 & -0.45 & 75.26 & -2.66 & 150.21 & 1.51 & 156.93 \\ 
  &$\hat{c}_{K}$ & 7.67 & 109.91 & 6.77 & 74.57 & 4.68 & 39.14 & 6.38 & 96.15 & 7.36 & 77.39 \\ 
  &$\hat{c}_{H}$ & 0.40 & 186.56 & 1.30 & 110.77 & 1.12 & 68.91 & -0.36 & 136.04 & 1.29 & 147.84 \\ \hline
  \multirow{6}{*}{$0.6$}& $\hat{c}$ & -0.35 & 32.29 & -0.18 & 15.91 & 0.01 & 7.38 & -0.53 & 23.77 & 0.06 & 18.20 \\ 
  &$\hat{c}_{B}$ & -2.75 & 43.66 & -2.04 & 21.86 & -1.76 & 10.88 & -2.45 & 30.09 & -1.89 & 23.40 \\ 
   &$\hat{c}_{G}$ & -5.87 & 108.83 & -3.71 & 54.12 & -2.30 & 26.28 & -3.94 & 52.16 & -2.82 & 45.86 \\ 
  &$\hat{c}_{E}$  & -2.29 & 161.37 & -0.67 & 118.08 & 0.16 & 74.77 & -2.40 & 147.73 & 1.22 & 146.35 \\ 
  &$\hat{c}_{K}$ & 7.57 & 119.52 & 6.43 & 74.85 & 5.08 & 42.21 & 5.93 & 99.47 & 7.62 & 77.68 \\ 
  &$\hat{c}_{H}$ & 0.40 & 174.95 & 0.78 & 82.95 & 0.60 & 42.69 & 0.16 & 100.73 & 0.25 & 127.29 \\ \hline
  \multirow{6}{*}{$0.8$}& $\hat{c}$ & -0.41 & 71.84 & -0.31 & 36.11 & 0.02 & 17.53 & -0.80 & 52.89 & 0.20 & 41.00 \\ 
  &$\hat{c}_{B}$ & -1.53 & 60.09 & -1.11 & 30.31 & -0.78 & 13.91 & -1.71 & 46.46 & -0.66 & 31.67 \\ 
   &$\hat{c}_{G}$ & -7.14 & 237.07 & -4.14 & 111.72 & -2.58 & 56.25 & -3.34 & 99.57 & -4.87 & 98.48 \\ 
  &$\hat{c}_{E}$  & -3.17 & 236.85 & -1.72 & 159.92 & -1.24 & 107.54 & -4.03 & 195.30 & 1.09 & 197.79 \\ 
  &$\hat{c}_{K}$ & 7.09 & 184.22 & 6.21 & 113.61 & 5.07 & 65.86 & 5.56 & 152.15 & 7.02 & 107.26 \\ 
  &$\hat{c}_{H}$ & -2.30 & 169.65 & -2.76 & 104.87 & -2.90 & 56.83 & -3.58 & 116.50 & -1.74 & 146.35 \\ \hline
\end{tabular}
\end{table}

\begin{table}[!htt]
\centering
\small
\caption{RB (\%) and MSE ($\times 100$) for point estimators of $J$ when the LLOD equals 15\% quantile of $F_0$} \label{point.J.lod}
\begin{tabular}{cccccccccccc}
  \hline
 &$(n_0, n_1)$& \multicolumn{2}{c}{$(50, 50)$} & \multicolumn{2}{c}{$(100, 100)$} & \multicolumn{2}{c}{$(200, 200)$}& \multicolumn{2}{c}{$(50, 150)$} & \multicolumn{2}{c}{$(150, 50)$}\\ 
  \hline
  $J$& & RB & MSE & RB & MSE& RB & MSE& RB & MSE& RB & MSE\\ \hline
  \multirow{4}{*}{$0.2$}& $\hat{J}$ & 5.92 & 0.61 & 2.49 & 0.31 & 1.11 & 0.14 & 4.62 & 0.40 & 2.85 & 0.42 \\ 
  &$\hat{J}_{B}$ & 9.59 & 0.64 & 4.75 & 0.32 & 2.85 & 0.15 & 7.38 & 0.41 & 5.85 & 0.43 \\ 
  &$\hat{J}_{G}$ & 7.58 & 0.69 & 2.67 & 0.36 & -0.34 & 0.17 & 5.99 & 0.46 & 3.50 & 0.50 \\ 
  &$\hat{J}_{E}$ & 40.17 & 1.29 & 26.63 & 0.64 & 17.72 & 0.30 & 33.04 & 0.86 & 32.69 & 0.89 \\ \hline
  \multirow{4}{*}{$0.4$}& $\hat{J}$ & 2.69 & 0.57 & 1.10 & 0.30 & 0.51 & 0.13 & 1.88 & 0.36 & 1.31 & 0.40 \\ 
  &$\hat{J}_{B}$ & 4.92 & 0.59 & 2.59 & 0.29 & 1.75 & 0.14 & 3.59 & 0.36 & 3.12 & 0.41 \\ 
  &$\hat{J}_{G}$ & 1.50 & 0.61 & -0.21 & 0.32 & -1.20 & 0.15 & 0.68 & 0.39 & -0.01 & 0.46 \\ 
  &$\hat{J}_{E}$ & 16.42 & 1.06 & 10.62 & 0.53 & 6.78 & 0.24 & 13.36 & 0.68 & 12.57 & 0.74 \\ \hline
  \multirow{4}{*}{$0.6$}& $\hat{J}$ & 1.74 & 0.46 & 0.6 & 0.24 & 0.29 & 0.11 & 1.04 & 0.27 & 0.85 & 0.34 \\ 
  &$\hat{J}_{B}$ & 3.17 & 0.44 & 1.70 & 0.22 & 1.19 & 0.11 & 2.16 & 0.27 & 2.09 & 0.32 \\ 
  &$\hat{J}_{G}$ & 0.11 & 0.47 & -0.64 & 0.26 & -1.04 & 0.12 & -0.56 & 0.29 & -0.35 & 0.38 \\ 
  &$\hat{J}_{E}$ & 8.89 & 0.76 & 5.38 & 0.37 & 3.71 & 0.19 & 7.14 & 0.48 & 6.73 & 0.56 \\ \hline
  \multirow{4}{*}{$0.8$}& $\hat{J}$ & 1.10 & 0.26 & 0.40 & 0.14 & 0.16 & 0.06 & 0.58 & 0.14 & 0.63 & 0.21 \\ 
  &$\hat{J}_{B}$ & 1.58 & 0.22 & 0.80 & 0.11 & 0.56 & 0.05 & 1.01 & 0.12 & 1.02 & 0.17 \\ 
  &$\hat{J}_{G}$ & -0.60 & 0.30 & -0.65 & 0.16 & -0.80 & 0.08 & -0.94 & 0.17 & -0.39 & 0.24 \\ 
  &$\hat{J}_{E}$ & 4.51 & 0.39 & 2.86 & 0.20 & 1.94 & 0.10 & 3.73 & 0.24 & 3.53 & 0.30 \\ \hline
\end{tabular}
\end{table}

\begin{table}[!htt]
\centering
\small
\caption{RB (\%) and MSE ($\times100$) for point estimators of $c$ when the LLOD equals 15\% quantile of $F_0$}
\label{point.c.lod}
\begin{tabular}{cccccccccccc}
  \hline
 &$(n_0, n_1)$& \multicolumn{2}{c}{$(50, 50)$} & \multicolumn{2}{c}{$(100, 100)$} & \multicolumn{2}{c}{$(200, 200)$}& \multicolumn{2}{c}{$(50, 150)$} & \multicolumn{2}{c}{$(150, 50)$}\\ 
  \hline
  $J$& & RB & MSE & RB & MSE& RB & MSE& RB & MSE& RB & MSE\\ \hline
  \multirow{4}{*}{$0.2$}& $\hat{c}$ & -1.58 & 60.04 & -0.86 & 24.25 & -0.44 & 10.84 & -1.53 & 38.67 & -0.60 & 66.67 \\ 
  &$\hat{c}_{B}$ & 1.04 & 148.39 & 0.49 & 51.29 & -0.30 & 24.87 & 0.29 & 66.35 & -0.07 & 67.79 \\ 
  &$\hat{c}_{G}$ & 6.30 & 168.58 & 8.65 & 81.29 & 8.51 & 48.7 & 6.72 & 90.4 & 9.00 & 102.15 \\ 
  &$\hat{c}_{E}$ & -0.65 & 257.51 & 0.49 & 170.8 & -0.76 & 116.8 & -0.79 & 202.38 & 1.89 & 223.79 \\ \hline
  \multirow{4}{*}{$0.4$}& $\hat{c}$ & -0.59 & 23.70 & -0.24 & 11.48 & -0.04 & 5.58 & -0.67 & 15.51 & -0.15 & 15.28 \\ 
  &$\hat{c}_{B}$ & -0.82 & 47.68 & -0.54 & 24.65 & -0.57 & 12.19 & -1.06 & 31.41 & -0.33 & 29.38 \\ 
  &$\hat{c}_{G}$ & 2.35 & 75.60 & 5.21 & 47.69 & 5.98 & 30.66 & 4.04 & 51.81 & 6.34 & 54.94 \\ 
  &$\hat{c}_{E}$ & -2.89 & 160.26 & -0.23 & 126.83 & -0.45 & 75.26 & -2.66 & 150.21 & 1.51 & 156.93 \\ \hline
 \multirow{4}{*}{$0.6$}& $\hat{c}$ & -0.39 & 32.90 & -0.22 & 16.35 & -0.01 & 7.71 & -0.63 & 23.81 & 0.01 & 19.68 \\ 
  &$\hat{c}_{B}$ & -1.05 & 48.30 & -0.81 & 23.92 & -0.64 & 11.46 & -1.48 & 31.25 & -0.27 & 26.75 \\ 
  &$\hat{c}_{G}$ & 0.18 & 91.92 & 2.33 & 51.66 & 3.77 & 29.23 & 2.04 & 52.36 & 3.34 & 50.15 \\ 
  &$\hat{c}_{E}$ & -2.29 & 161.37 & -0.67 & 118.08 & 0.16 & 74.77 & -2.40 & 147.73 & 1.22 & 146.35 \\ \hline
  \multirow{4}{*}{$0.8$}& $\hat{c}$ & -0.45 & 72.29 & -0.34 & 36.11 & 0.02 & 17.60 & -0.87 & 53.45 & 0.17 & 41.41 \\ 
  &$\hat{c}_{B}$ & -0.65 & 66.17 & -0.46 & 32.77 & -0.17 & 15.11 & -1.24 & 48.59 & 0.25 & 35.50 \\ 
  &$\hat{c}_{G}$ & -3.53 & 193.47 & -0.26 & 92.54 & 1.59 & 52.64 & 0.28 & 91.43 & -0.43 & 99.37 \\ 
  &$\hat{c}_{E}$ & -3.17 & 236.85 & -1.72 & 159.92 & -1.24 & 107.54 & -4.03 & 195.30 & 1.09 & 197.79 \\ \hline
\end{tabular}
\end{table}


When there is no LLOD, major observations from Tables~\ref{point.J.no}--\ref{point.c.no} can be summarized as follows. 
For estimating the Youden index $J$, 
the estimators $\hat J$, $\hat J_B$, $\hat J_G$, and $\hat J_K$ have comparable performance in terms of MSE, 
which are uniformly better than $\hat J_E$ and $\hat J_H$. 
When sample sizes are small, the kernel-based estimator $\hat J_K$ may have slightly smaller MSE than $\hat J$;
when the sample size increases, our proposed estimator $\hat J$ becomes more efficient in terms of MSE.  
This is in line with our discussion after Theorem 1.
The ECDF-based estimator $\hat J_E$ has the largest RBs and MSEs in almost all the cases.
We also notice that  when $J =0.2$, the RBs of $\hat J_B$, $\hat J_G$, and $\hat{J}_H$ 
have greater than 5\% RBs, which may not be acceptable, especially when one of $n_0$ and $n_1$ is small.

For estimating the optimal cutoff point $c$, our proposed estimator $\hat c$ outperforms other estimators significantly for the majority of cases. 
The parametric estimator $\hat{c}_B$ is most competitive. 
It has larger MSEs than $\hat c$ when $J=0.2,0.4, 0.6$
and has slightly smaller MSEs  than $\hat c$ when $J=0.8$. 
Among the other four estimators, the estimator $\hat{c}_E$ has the worst performance and $\hat c_G$ shows the best performance in most cases.
The performances of $\hat c_K$ and $\hat c_H$ are mixed. There is no obvious trend that one dominates the other one.

When the LLOD equals 15\% quantile of $F_0$,
Tables \ref{point.J.lod}--\ref{point.c.lod} show that 
the general trend for comparing the proposed method with 
the Box-Cox method, ROC-GLM method, and ECDF-based method
are similar to the case when there is no LLOD. 
It is worthy mentioning that  as the LLOD increases, the MSEs of all estimators increase, due to the loss of information under censoring. 
The estimation of the optimal cutoff point $c$ is more sensitive to the increase of LLOD especially when $J$ is small. 

\subsection{Comparison for confidence intervals} 

We now examine the behaviour of the 95\% CIs of $(J,c)$.  
The performance of a CI is evaluated by the coverage probability (CP) in percentage and the average length (AL) computed as 
\begin{equation*}
    CP = \frac{1}{B}\sum_{b=1}^{B}I(a_0 \in \mathcal{I}^{(b)}) \times 100,~~~ AL = \frac{1}{B}\sum_{b=1}^{B}|\mathcal{I}^{(b)}|,
\end{equation*}
where $\mathcal{I}^{(b)}$ is the CI of the interested quantity computed from the $b$th simulation run, and $|\cdot|$ is the length of the CI.
The simulation results are presented in Tables \ref{CI.J.no}--\ref{CI.c.lod}.

\begin{table}[!htt]
\centering
\small
\caption{CP (\%) and AL for CIs of $J$ when there is no LLOD}
\label{CI.J.no}
\begin{tabular}{cccccccccccc}
  \hline
 &$(n_0, n_1)$& \multicolumn{2}{c}{$(50, 50)$} & \multicolumn{2}{c}{$(100, 100)$} & \multicolumn{2}{c}{$(200, 200)$}& \multicolumn{2}{c}{$(50, 150)$} & \multicolumn{2}{c}{$(150, 50)$}\\ 
  \hline
   $J$& & CP & AL & CP & AL& CP & AL& CP & AL& CP & AL\\\hline
  \multirow{5}{*}{$0.2$}& $\mathcal{I}_J$ & 94.1 & 0.32 & 94.9 & 0.22 & 96.4 & 0.15 & 93.6 & 0.25 & 94.7 & 0.26 \\ 
  &$\mathcal{I}_{JB}$ & 94.5 & 0.30 & 93.9 & 0.21 & 95.2 & 0.15 & 92.7 & 0.24 & 94.6 & 0.25 \\ 
  &$\mathcal{I}_{JG}$ & 92.6 & 0.29 & 93.7 & 0.21 & 94.7 & 0.15 & 91.9 & 0.23 & 93.7 & 0.25 \\ 
  &$\mathcal{I}_{JE}$  & 70.2 & 0.30 & 72.1 & 0.22 & 78.0 & 0.16 & 66.7 & 0.24 & 72.6 & 0.26 \\ 
  &$\mathcal{I}_{JK}$ & 93.6 & 0.30 & 93.7 & 0.23 & 94.7 & 0.16 & 92.3 & 0.25 & 93.7 & 0.26 \\
  &$\mathcal{I}_{JH}$  & 92.1 & 0.31 & 94.0 & 0.23 & 94.1 & 0.17 & 91.1 &0.25  &93.3 & 0.27  \\ \hline
  \multirow{5}{*}{$0.4$}& $\mathcal{I}_J$ & 95.4 & 0.28 & 94.9 & 0.20 & 96.0 & 0.14 & 93.3 & 0.22 & 94.3 & 0.24 \\ 
  &$\mathcal{I}_{JB}$ & 95.1 & 0.28 & 94.2 & 0.20 & 95.3 & 0.14 & 93.1 & 0.22 & 94.2 & 0.24 \\ 
  &$\mathcal{I}_{JG}$ & 93.1 & 0.29 & 93.6 & 0.20 & 94.7 & 0.14 & 91.6 & 0.22 & 92.9 & 0.25 \\ 
  &$\mathcal{I}_{JE}$ & 79.2 & 0.30 & 80.5 & 0.22 & 83.5 & 0.16 & 75.1 & 0.23 & 80.4 & 0.26 \\ 
  &$\mathcal{I}_{JK}$ & 93.7 & 0.29 & 93.1 & 0.21 & 95.1 & 0.15 & 92.7 & 0.23 & 93.5 & 0.24 \\
  &$\mathcal{I}_{JH}$ & 93.7 & 0.31 & 95.0 & 0.22 &94.7  & 0.16 & 93.3 & 0.23 & 96.5 &0.28  \\ \hline
  \multirow{5}{*}{$0.6$}& $\mathcal{I}_J$ & 95.6 & 0.25 & 94.3 & 0.18 & 95.0 & 0.13 & 94.2 & 0.19 & 94.0 & 0.22 \\ 
  &$\mathcal{I}_{JB}$ & 95.8 & 0.25 & 94.2 & 0.18 & 95.1 & 0.13 & 93.8 & 0.19 & 94.2 & 0.22 \\ 
  &$\mathcal{I}_{JG}$ & 93.3 & 0.26 & 93.5 & 0.18 & 94.8 & 0.13 & 92.5 & 0.19 & 92.6 & 0.23 \\ 
  &$\mathcal{I}_{JE}$ & 79.4 & 0.26 & 80.7 & 0.19 & 85.3 & 0.14 & 75.7 & 0.20 & 81.8 & 0.23 \\ 
  &$\mathcal{I}_{JK}$ & 94.4 & 0.24 & 94.1 & 0.17 & 95.4 & 0.13 & 92.4 & 0.19 & 93.5 & 0.21 \\ 
  &$\mathcal{I}_{JH}$& 94.0 & 0.30 & 95.3 & 0.21 & 95.3 & 0.15 &92.6  &0.21  & 95.8 & 0.27 \\ \hline
  \multirow{5}{*}{$0.8$}& $\mathcal{I}_J$ & 95.1 & 0.20 & 94.0 & 0.14 & 95.0 & 0.10 & 95.1 & 0.14 & 95.9 & 0.18 \\ 
  &$\mathcal{I}_{JB}$ & 96.0 & 0.18 & 94.1 & 0.13 & 95.3 & 0.09 & 94.1 & 0.13 & 95.2 & 0.16 \\ 
  &$\mathcal{I}_{JG}$ & 92.7 & 0.20 & 94.0 & 0.14 & 95.2 & 0.10 & 93.3 & 0.14 & 92.2 & 0.18 \\ 
  &$\mathcal{I}_{JE}$ & 81.6 & 0.19 & 85.5 & 0.14 & 85.9 & 0.10 & 76.1 & 0.14 & 83.5 & 0.17 \\ 
  &$\mathcal{I}_{JK}$ & 95.0 & 0.17 & 91.8 & 0.13 & 91.0 & 0.09 & 93.6 & 0.13 & 92.3 & 0.15 \\
  &$\mathcal{I}_{JH}$ &  86.9 & 0.23&   95.8& 0.17 & 95.5 & 0.12 & 92.8 &0.17  & 92.2&	0.22\\
   \hline
\end{tabular}
\end{table}

\begin{table}[!htt]
\centering
\small
\caption{CP (\%) and AL for CIs of $c$ when there is no LLOD}
\label{CI.c.no} 
\begin{tabular}{cccccccccccc}
  \hline
 &$(n_0, n_1)$& \multicolumn{2}{c}{$(50, 50)$} & \multicolumn{2}{c}{$(100, 100)$} & \multicolumn{2}{c}{$(200, 200)$}& \multicolumn{2}{c}{$(50, 150)$} & \multicolumn{2}{c}{$(150, 50)$}\\ 
  \hline
  $J$& & CP & AL & CP & AL& CP & AL& CP & AL& CP & AL\\\hline
  \multirow{5}{*}{$0.2$}& $\mathcal{I}_c$ & 94.1 & 1.31 & 94.1 & 0.94 & 95.2 & 0.67 & 94.2 & 1.11 & 94.8 & 1.05 \\ 
  &$\mathcal{I}_{cB}$ & 92.8 & 7.26 & 93.9 & 2.64 & 93.8 & 1.77 & 93.7 & 3.25 & 92.6 & 3.25 \\ 
  &$\mathcal{I}_{cG}$ & 94.0 & 4.50 & 94.7 & 3.02 & 94.3 & 2.07 & 93.7 & 3.43 & 94.8 & 3.47 \\ 
  &$\mathcal{I}_{cE}$  & 97.0 & 5.36 & 96.5 & 4.65 & 96.8 & 3.89 & 95.6 & 4.83 & 97.2 & 5.05 \\ 
  &$\mathcal{I}_{cK}$ & 95.6 & 7.02 & 94.3 & 5.70 & 95.0 & 3.89 & 95.2 & 6.68 & 92.6 & 6.36 \\ 
  &$\mathcal{I}_{cH}$ & 97.2& 5.98 & 97.9  & 4.79 &95.8  & 3.67 & 96.4 & 4.59 &95.4 &5.80  \\\hline
 \multirow{5}{*}{$0.4$}& $\mathcal{I}_c$ & 93.9 & 1.63 & 94.5 & 1.16 & 95.6 & 0.82 & 93.3 & 1.40 & 94.5 & 1.27 \\ 
 &$\mathcal{I}_{cB}$  & 92.0 & 2.38 & 92.0 & 1.72 & 92.5 & 1.23 & 92.5 & 1.98 & 92.1 & 1.83 \\ 
 &$\mathcal{I}_{cG}$ & 91.5 & 3.13 & 94.0 & 2.30 & 93.4 & 1.63 & 92.2 & 2.36 & 93.4 & 2.31 \\ 
  &$\mathcal{I}_{cE}$  & 95.1 & 4.52 & 96.6 & 3.80 & 95.0 & 3.18 & 94.2 & 4.02 & 96.8 & 4.32 \\ 
  &$\mathcal{I}_{cK}$ & 93.9 & 4.23 & 93.2 & 2.91 & 93.4 & 2.19 & 94.3 & 3.55 & 92.5 & 3.13 \\ 
  &$\mathcal{I}_{cH}$ & 97.0 & 4.87 & 95.6 & 3.83 & 94.6 &2.92  & 94.4 & 3.97 &  97.1& 4.58 \\ \hline
  \multirow{5}{*}{$0.6$}& $\mathcal{I}_c$ & 93.2 & 2.13 & 94.4 & 1.52 & 95.0 & 1.08 & 93.4 & 1.85 & 94.0 & 1.63 \\ 
  &$\mathcal{I}_{cB}$  & 91.5 & 2.36 & 92.1 & 1.70 & 92.2 & 1.21 & 91.4 & 1.99 & 92.3 & 1.76 \\ 
  &$\mathcal{I}_{cG}$ & 90.2 & 3.57 & 90.8 & 2.59 & 92.4 & 1.86 & 90.9 & 2.57 & 91.1 & 2.44 \\ 
  &$\mathcal{I}_{cE}$  & 94.5 & 4.43 & 95.7 & 3.76 & 97.0 & 3.05 & 92.2 & 3.96 & 95.9 & 4.21 \\ 
  &$\mathcal{I}_{cK}$ & 92.7 & 3.60 & 90.9 & 2.73 & 90.0 & 2.04 & 92.7 & 3.31 & 87.8 & 2.65 \\ 
  &$\mathcal{I}_{cH}$& 96.7 & 4.80 &96.6  & 3.76 &  96.4& 2.78 & 94.7 & 3.80 & 97.6 &4.50  \\  \hline
  \multirow{5}{*}{$0.8$}& $\mathcal{I}_c$ & 92.9 & 3.25 & 94.3 & 2.31 & 95.0 & 1.65 & 93.8 & 2.78 & 93.7 & 2.46 \\ 
  &$\mathcal{I}_{cB}$  & 91.9 & 2.87 & 93.3 & 2.03 & 94.5 & 1.44 & 90.6 & 2.49 & 93.8 & 2.14 \\ 
  &$\mathcal{I}_{cG}$ & 87.8 & 5.18 & 91.5 & 3.81 & 90.0 & 2.71 & 92.0 & 3.65 & 86.6 & 3.27 \\ 
  &$\mathcal{I}_{cE}$  & 88.8 & 4.93 & 93.7 & 4.28 & 93.9 & 3.55 & 84.5 & 4.42 & 94.5 & 4.78 \\ 
  &$\mathcal{I}_{cK}$ & 91.3 & 4.34 & 90.9 & 3.29 & 89.7 & 2.49 & 92.1 & 4.06 & 88.8 & 3.16 \\  
  &$\mathcal{I}_{cH}$&95.6 & 4.74  & 97.3  & 3.99 &94.8  &2.91  & 93.6 &3.97  & 98.7 &4.74 \\ \hline
\end{tabular}
\end{table}

\begin{table}[!htt]
\centering
\small
\caption{CP (\%) and AL for CIs of  $J$ when the LLOD equals 15\% quantile of $F_0$}\label{CI.J.lod}
\begin{tabular}{cccccccccccc}
  \hline
 &$(n_0, n_1)$& \multicolumn{2}{c}{$(50, 50)$} & \multicolumn{2}{c}{$(100, 100)$} & \multicolumn{2}{c}{$(200, 200)$}& \multicolumn{2}{c}{$(50, 150)$} & \multicolumn{2}{c}{$(150, 50)$}\\ 
  \hline
   $J$& & CP & AL & CP & AL& CP & AL& CP & AL& CP & AL\\\hline
   \multirow{4}{*}{$0.2$}& $\mathcal{I}_J$ & 93.4 & 0.31 & 94.2 & 0.21 & 96.0 & 0.15 & 93.8 & 0.25 & 94.3 & 0.25 \\ 
  &$\mathcal{I}_{JB}$ & 94.3 & 0.30 & 93.9 & 0.21 & 95.2 & 0.15 & 92.6 & 0.24 & 94.5 & 0.25 \\ 
  &$\mathcal{I}_{JG}$ & 93.2 & 0.31 & 94.0 & 0.23 & 94.7 & 0.16 & 92.9 & 0.25 & 93.2 & 0.26 \\ 
  &$\mathcal{I}_{JE}$ & 64.8 & 0.28 & 68.2 & 0.21 & 73.4 & 0.15 & 61.4 & 0.23 & 65.0 & 0.24 \\ \hline
  \multirow{4}{*}{$0.4$}& $\mathcal{I}_J$ & 95.3 & 0.28 & 94.9 & 0.20 & 95.8 & 0.15 & 93.7 & 0.22 & 94.5 & 0.24 \\ 
  &$\mathcal{I}_{JB}$ & 95.4 & 0.28 & 94.0 & 0.20 & 95.4 & 0.14 & 93.3 & 0.22 & 94.6 & 0.24 \\ 
  &$\mathcal{I}_{JG}$ & 94.6 & 0.30 & 94.8 & 0.22 & 95.2 & 0.15 & 92.7 & 0.24 & 93.3 & 0.26 \\ 
  &$\mathcal{I}_{JE}$ & 77.4 & 0.28 & 78.5 & 0.21 & 81.7 & 0.15 & 73.8 & 0.22 & 78.1 & 0.25 \\ \hline
  \multirow{4}{*}{$0.6$}& $\mathcal{I}_J$& 95.7 & 0.26 & 94.6 & 0.18 & 94.7 & 0.13 & 94.0 & 0.19 & 94.2 & 0.22 \\ 
  &$\mathcal{I}_{JB}$ & 95.7 & 0.25 & 94.2 & 0.18 & 94.7 & 0.13 & 93.8 & 0.19 & 94.2 & 0.22 \\ 
  &$\mathcal{I}_{JG}$ & 94.3 & 0.27 & 93.9 & 0.19 & 94.5 & 0.14 & 92.9 & 0.20 & 93.4 & 0.23 \\ 
  &$\mathcal{I}_{JE}$ & 77.6 & 0.25 & 79.3 & 0.19 & 82.8 & 0.14 & 74.3 & 0.19 & 80.3 & 0.22 \\ \hline
  \multirow{4}{*}{$0.8$}& $\mathcal{I}_J$ & 95.5 & 0.20 & 94.2 & 0.14 & 95.3 & 0.10 & 94.9 & 0.14 & 95.7 & 0.18 \\ 
  &$\mathcal{I}_{JB}$ & 96.2 & 0.18 & 94.8 & 0.13 & 94.6 & 0.09 & 94.6 & 0.13 & 95.2 & 0.16 \\ 
  &$\mathcal{I}_{JG}$ & 93.7 & 0.21 & 93.5 & 0.15 & 95.4 & 0.11 & 94.3 & 0.15 & 93.8 & 0.19 \\ 
  &$\mathcal{I}_{JE}$ & 80.2 & 0.18 & 85.3 & 0.14 & 86.4 & 0.10 & 75.1 & 0.14 & 83.1 & 0.17 \\ \hline
   \end{tabular}
\end{table}

\begin{table}[!htt]
\centering
\small
\caption{CP (\%) and AL for CIs of $c$ when the LLOD equals 15\% quantile of $F_0$}\label{CI.c.lod}
\begin{tabular}{cccccccccccc}
  \hline
 &$(n_0, n_1)$& \multicolumn{2}{c}{$(50, 50)$} & \multicolumn{2}{c}{$(100, 100)$} & \multicolumn{2}{c}{$(200, 200)$}& \multicolumn{2}{c}{$(50, 150)$} & \multicolumn{2}{c}{$(150, 50)$}\\ 
  \hline
   $J$& & CP & AL & CP & AL& CP & AL& CP & AL& CP & AL\\\hline
   \multirow{4}{*}{$0.2$}& $\mathcal{I}_c$ & 97.5 & 5.97 & 97.5 & 1.95 & 96.5 & 1.26 & 96.9 & 2.73 & 96.8 & 4.02 \\ 
  &$\mathcal{I}_{cB}$ & 83.6 & 5.06 & 83.7 & 2.10 & 81.8 & 1.43 & 77.4 & 2.35 & 84.8 & 2.74 \\ 
  &$\mathcal{I}_{cG}$ & 95.3 & 4.83 & 93.1 & 3.27 & 89.5 & 2.23 & 91.1 & 3.80 & 91.9 & 3.77 \\ 
  &$\mathcal{I}_{cE}$ & 96.6 & 5.07 & 95.4 & 4.41 & 96.9 & 3.75 & 94.6 & 4.55 & 96.4 & 4.82 \\ \hline
  \multirow{4}{*}{$0.4$}& $\mathcal{I}_c$ & 95.3 & 1.87 & 95.0 & 1.32 & 95.7 & 0.93 & 94.0 & 1.53 & 95.5 & 1.52 \\ 
  &$\mathcal{I}_{cB}$ & 86.8 & 2.19 & 87.9 & 1.58 & 87.8 & 1.12 & 83.5 & 1.58 & 91.3 & 1.84 \\ 
  &$\mathcal{I}_{cG}$ & 94.3 & 3.20 & 92.7 & 2.36 & 87.3 & 1.68 & 91.2 & 2.48 & 92.0 & 2.41 \\ 
  &$\mathcal{I}_{cE}$ & 95.1 & 4.43 & 95.9 & 3.74 & 95.3 & 3.14 & 93.6 & 3.93 & 97.3 & 4.27 \\ \hline
 \multirow{4}{*}{$0.6$}& $\mathcal{I}_c$ & 93.4 & 2.17 & 94.3 & 1.55 & 95.3 & 1.10 & 93.5 & 1.85 & 94.3 & 1.70 \\ 
  &$\mathcal{I}_{cB}$ & 91.0 & 2.39 & 91.3 & 1.71 & 92.1 & 1.22 & 87.8 & 1.81 & 93.0 & 1.91 \\ 
  &$\mathcal{I}_{cG}$ & 94.0 & 3.59 & 94.2 & 2.62 & 92.3 & 1.88 & 93.7 & 2.65 & 93.6 & 2.51 \\ 
  &$\mathcal{I}_{cE}$ & 94.4 & 4.40 & 95.5 & 3.72 & 96.5 & 3.04 & 91.6 & 3.94 & 95.9 & 4.17 \\ \hline
  \multirow{4}{*}{$0.8$}& $\mathcal{I}_c$& 92.9 & 3.24 & 94.4 & 2.31 & 94.9 & 1.65 & 93.2 & 2.78 & 93.7 & 2.48 \\ 
  &$\mathcal{I}_{cB}$& 92.8 & 2.99 & 94.9 & 2.14 & 94.7 & 1.52 & 91.1 & 2.51 & 94.2 & 2.27 \\ 
  &$\mathcal{I}_{cG}$& 92.5 & 5.28 & 95.0 & 3.79 & 94.8 & 2.71 & 94.2 & 3.74 & 92.8 & 3.39 \\ 
  &$\mathcal{I}_{cE}$& 88.5 & 4.90 & 93.5 & 4.26 & 94.3 & 3.54 & 84.7 & 4.39 & 94.3 & 4.76 \\ \hline
   \end{tabular}
\end{table}

We first summarize the findings from the CIs for the Youden index $J$. 
We can see that the CPs of $\mathcal{I}_{JE}$ are not acceptable regardless of the value of LLODs. 
The proposed CI  and the CI based on the Box-Cox method, $\mathcal{I}_{J}$ and $\mathcal{I}_{JB}$,
have the comparable and most stable performance in almost all cases.
The GLM-ROC based CI, $\mathcal{I}_{JG}$, perform quite well overall, but  may have undercoverage in some cases. 
When there is no LLOD, the two confidence intervals $\mathcal{I}_{JK}$ and $\mathcal{I}_{JH}$ have a similar issue for $\mathcal{I}_{JG}$ with  undercoverage problems. 

We next discuss the findings from the CIs for the optimal cutoff point $c$.  
When there is no LLOD, the proposed CI \ $\mathcal{I}_{c}$ \ has the most stable performance
and its CPs are reasonably close to 95\% in almost all scenarios. 
The CPs of \ $\mathcal{I}_{cE}$ \ fluctuate around the nominal level 95\% while 
undercoverage problems are associated with the other four CIs $\mathcal{I}_{cB}$, $\mathcal{I}_{cG}$, $\mathcal{I}_{cK}$,  and $\mathcal{I}_{cH}$. 
When there is a fixed and finite LLOD, the ALs of all CIs increase. 
The proposed CI \ $\mathcal{I}_{c}$ \ and the ECDF-based CI \ $\mathcal{I}_{cE}$ \ tend to have an issue with overcoverage,
while the CI based on the Box-Cox method has severe undercoverage problem 
and the GLM-ROC based CI \ $\mathcal{I}_{cG}$ \ also has the same issue for some cases. 
When $J=0.4,0.6,0.8$, our proposed CI \ $\mathcal{I}_{c}$ \ becomes quite stable in almost all cases.   
The performance of $\mathcal{I}_{cB}$ improves as $J$ increases. 
The CPs of $\mathcal{I}_{cG}$ are reasonably close to the nominal level. 
However, $\mathcal{I}_{cG}$ has longer ALs compared to $\mathcal{I}_{c}$.


\section{Real Data Analysis}
\label{sec6}

In this section, we illustrate the performance of the proposed method by analyzing a dataset on Duchenne Muscular Dystrophy (DMD).  The DMD is a genetic disorder characterized by progressive muscle degeneration and weakness. 
A particular gene on the X chromosome, when mutated, leads to DMD. 
This disease is transmitted from a mother to her children genetically. 
Affected male offsprings usually develop the disease and die at a young age while the mutated gene does not affect the health of female offsprings. Therefore, detection of potential affected females is of great interest.

\cite{percy1982duchenne} pointed out that carriers of DMD tend to exhibit high levels of certain biomarkers even though they do not show any symptoms. \cite{andrews2012data} collected the complete data of four biomarkers, namely, creatine kinase (CK), hemopexin (H), lactate dehydroginase (LD), and pyruvate kinase (PK),  from the blood serum samples of a healthy group of people ($n_0$ = 127) and a group of carriers ($n_1$ = 67). Our goal is to choose the most appropriate biomarker to distinguish healthy individuals from diseased ones. 

We choose $\mathbf{q}(x)=x$ in the proposed method for each biomarker, 
which is equivalent to assuming a logistic regression model for an individual's disease status and the biomarker \citep{qin1997goodness}. 
Table~\ref{drm.test} presents \cite{qin1997goodness}'s test statistics along with the p-values 
for the goodness of fit of the DRM in \eqref{drm} with $\mathbf{q}(x)=x$.  
It shows that for each biomarker, the data does not provide evidence to reject the DRM in \eqref{drm} with $\mathbf{q}(x)=x$. 

\begin{table}[!htt]
\centering
\small
\caption{ \cite{qin1997goodness}'s test statistics and their p-values when $\mathbf{q}(x)=x$.\label{drm.test}}
\begin{tabular}{c|cccc}
  \hline
  Biomarker & CK & LD & PH& H  \\ 
  \hline
Test statistic  & 0.211& 0.377 & 0.346 & 0.339 \\
  P-value & 0.912 & 0.291& 0.507 & 0.676\\
   \hline
\end{tabular}
\end{table}

Table \ref{realdata.results} provides the point estimates and the CIs (in parentheses) from our proposed method and all the competitive methods listed in Section \ref{sec5}. 
As we can see, for all biomarkers, the point estimates of Youden index are similar for all methods: they differ only in the second digit. 
For the CIs of the Youden index, the methods with $\hat J$, $\hat J_B$, $\hat J_G$, and $\hat J_K$ have similar performances for all biomarkers; 
the CIs with $\hat J_E$ and $\hat J_H$ tend to be wider than other four methods. 
For the optimal cutoff point, the point estimates have substantial differences, especially for the biomarker LD, compared with the estimates of the Youden index. 
For all biomarkers, the proposed method has the shortest CIs, while the ECDF-based method  and HCNS method tend to have the widest CIs. 
The performances of other three CIs are mixed: the CI based on the Box-Cox method has shorter length for biomarkers CK and LD, 
while the CIs based on GLM-ROC and kernel methods have shorter length for biomarkers PK and H. 
Furthermore, we find that the biomarker CK gives the largest estimated Youden index which is around 0.6. 
Therefore, the biomarker CK performs the best among these four biomarkers
to distinguish the diseased individuals and the healthy ones. The estimated optimal cutoff point for the biomarker CK using our proposed method is 61.13 with the 95\% CI being $(54.59,67.68)$.

\begin{table}[!htt]
\centering
\small
\caption{Estimation of the Youden index and the optimal cutoff point with the DMD dataset}
\label{realdata.results}
    \begin{tabular}{ccccc}
    \hline
    &CK&LD&PK&H\\ \hline
  $\hat{J}$ & 0.59 (0.48, 0.69 ) & 0.55 (0.45, 0.65) & 0.49 (0.38, 0.59) & 0.36 (0.26, 0.48) \\ 
  $\hat{J}_B$ & 0.62 (0.51, 0.70) & 0.56 (0.46, 0.66) & 0.48 (0.37, 0.58) & 0.37 (0.26, 0.48) \\ 
  $\hat{J}_G$ & 0.60 (0.50, 0.71) & 0.57 (0.47, 0.68) & 0.48 (0.38, 0.61) & 0.39 (0.29, 0.50) \\ 
  $\hat{J}_E$ & 0.61 (0.52, 0.73) & 0.58 (0.50, 0.72) & 0.51 (0.42, 0.65) & 0.42 (0.34, 0.57) \\ 
  $\hat{J}_K$ & 0.59 (0.51, 0.67) & 0.55 (0.45, 0.66) & 0.47 (0.37, 0.58) & 0.37 (0.25, 0.49) \\ 
  $\hat{J}_H$ & 0.61 (0.52, 0.80) & 0.57 (0.46, 0.70) & 0.48 (0.35, 0.62) & 0.40 (0.31, 0.56) \\
   \hline
  $\hat{c}$ & 61.13 (54.59, 67.68) & 198.56 (190.34, 206.78) & 15.54 (14.65, 16.43) & 87.74 (86.09, 89.39) \\ 
  $\hat{c}_B$ & 58.01 (51.17, 65.42) & 200.01 (188.99, 209.41) & 16.56 (14.83, 18.24) & 86.73 (83.59, 89.35) \\ 
  $\hat{c}_G$ & 55.60 (48.83, 68.41) & 197.54 (183.47, 211.64) & 15.81 (14.58, 16.79) & 85.25 (82.31, 87.90) \\ 
  $\hat{c}_E$ & 56.00 (43.00, 75.00) & 187.00 (181.00, 232.00) & 16.60 (14.00, 18.20) & 87.20 (80.50, 88.50) \\ 
  $\hat{c}_K$ & 73.36 (54.15, 79.16) & 202.32 (188.31, 216.94) & 17.22 (15.87, 18.28) & 85.52 (82.84, 88.36) \\ 
  $\hat{c}_H$ & 52.02 (43.01, 68.50) & 202.92 (179.20, 221.22) &  14.37 (12.34, 18.05)&   82.90 (80.26, 92.10)\\
  \hline
    \end{tabular}
\end{table}

\section{Concluding Remarks}
\label{sec7}

In this paper, we propose to link the distributions of the biomarkers in the diseased and healthy groups 
via the DRM \eqref{drm}. 
Based on this model, we obtain the maximum empirical likelihood estimators of the Youden index and the corresponding optimal cutoff point. 
We further establish the asymptotic normality of the estimators, which enables us to construct valid CIs for the Youden index and the corresponding optimal cutoff point. 
The proposed method covers cases without a LLOD and also cases with a  fixed and finite LLOD. 
Simulation studies and a real data application demonstrate the advantages of our proposed method over existing methods. 

One problem arising from the simulation studies is that the proposed confidence interval $\mathcal{I}_c$ for the optimal cutoff point could have under/over coverage issues under certain scenarios, especially when there is a fixed and finite LLOD and one of the sample sizes is small. A possible alternative approach is to consider the EL ratio based CI for $c$.  Another problem is related to the real data application where there are multiple biomarkers. \cite{yin2014optimal} studied the optimal linear combination of multiple biomarkers based on the Youden index. 
We can first use the DRM to link multiple biomarkers, then construct a derived optimal linear combination of the biomarkers and find the optimal cutoff point based on the derived biomarker.
Both research problems are currently under investigation. 

We conclude the paper with some discussion on the choice of $\mathbf{q}(x)$.
To use the proposed method, we need to specify $\mathbf{q}(x)$ in advance. 
If the practitioners believe that a logistic regression model is adequate to describe the relationship between the individual's disease status and the biomarker, then they can use the DRM (\ref{drm}) with $\mathbf{q}(x)=x$.
If the practitioners believe that gamma distributions or lognormal distributions provide good fittings to  the biomarkers in the healthy and diseased groups,  then they can use the semiparametric DRM (\ref{drm}) with $\mathbf{q}(x)=(x,\log x)^T$ or $(\log x,\log^2 x)^T$ instead of a parametric model to achieve robustness of inferences.   
The DRM (\ref{drm}) with a particular choice of $\mathbf{q}(x)$ can be further checked by the goodness of fit test discussed in  \cite{qin1997goodness}.
We have implemented our proposed method along with \cite{qin1997goodness}'s test for some commonly used $\mathbf{q}(x)$
 in an R package \texttt{YoudenDRM}. It is available upon request. 
 However, if the practitioners do not have any prior belief or information on the distributions of the biomarkers in the healthy and diseased groups, then a nonparametric method such as the kernel-based method and the HCNS method may be preferable.

%

\section*{Appendix: Regularity Conditions and Proofs}

\subsection*{A.1 \ Regularity conditions}

The asymptotic properties of $( \hat J, \hat c)$ rely on the following regularity conditions. 

\vspace{1ex}
\noindent
{\bf C1}. For any $\epsilon >0$,  $  J_{\epsilon} =  \sup\limits_{|x-c_0|\geq \epsilon } \{F_0(x) - F_1(x)\} <  J_0.$

\noindent
{\bf C2}. The first and second derivatives of  $F_0(x)$ and $F_1(x)$  are continuous  in the  neighbourhood of $c_0$, with  $F'_0(c_0)-F_1'(c_0)=0$ and $F''_0(c_0)-F_1''(c_0)<0$.

\noindent
{\bf C3}. The total sample size $n=n_0+n_1\to \infty$ and $\rho=n_1/n_0$ remains a constant. 

\noindent
{\bf C4}. The two CDFs $F_0$ and $F_1$ satisfy the DRM (\ref{drm}) with a true parameter value $\boldsymbol{\theta}_0$ 
and  
$\int_{r}^\infty \exp\{\boldsymbol{\theta}^T\mathbf{Q}(x)\}dF_0 < \infty$ in a neighborhood of $\boldsymbol{\theta}_0$, and  $\int_{r}^\infty \mathbf{Q}(x)\mathbf{Q}(x)^T dF_0(x)$ is positive definite.

Condition C1 is from \cite{hsieh1996nonparametric}, which ensures $c_0$ is unique. 
Condition C2 comes from the definitions of the Youden index and its corresponding optimal cutoff point. 
Conditions C3 and C4 guarantee that the asymptotic results in \cite{cai2018empirical} can be applied. 

%

\subsection*{A.2 \ Some preparations}
This section serves as preparations for the proof of Theorem 1. 
We first introduce some further notation. 
Let 
$$
H(x) = F_0(x)-F_1(x),~~
 \hat{H}(x) = \hat{F}_0(x) - \hat{F}_1(x).$$ 
Then $J_0=H(c_0)$ and $\hat J= \hat H(\hat c)$.  
Further let 
$$
\Delta_{n0}=\sup_{x\geq r}|\hat F_0(x)-F_0(x) |,~~
\Delta_{n1}=\sup_{x\geq r}|\hat F_1(x)-F_0(x) |,~~
 \Delta_{n}=\sup_{x\geq r}| \hat{H} (x)  |. 
$$
Following the proof of Lemma 3 in \cite{cai2018empirical}, we have $\Delta_{n0}=O_p(n^{-1/2})$ and $\Delta_{n1}=O_p(n^{-1/2})$. 
Hence  $\Delta_{n}=O_p(n^{-1/2})$. 

We can establish the consistency of $\hat c$ and argue that,  
with the probability goes to 1, the estimator $\hat c$ is the solution to $\hat{\boldsymbol{\theta}}^T\mathbf{Q}(x) = 0$.

\begin{lemma}
\label{lemma.consistency}
Assume Conditions C1--C4 are satisfied. Then, as $n \to \infty$, we have 
\begin{equation}
\label{lemma.result1}
\hat{c} \xrightarrow{p} c_0~~\mbox{in probability}
\end{equation}
 and 
\begin{equation}
\label{lemma.result2}
P\left(
   \hat{\boldsymbol{\theta}}^T\mathbf{Q}(\hat c) = 0
   \right)\to 1.
\end{equation}
\end{lemma}

\begin{proof}{Proof}{}
For \eqref{lemma.result1}, 
it is sufficient to show that for any $0<\epsilon<c_0-r$, 
\begin{eqnarray}
    \label{consis1}
    \lim \limits_{n \to \infty} P(\hat{c} > c_0+\epsilon) = 0,\\
        \label{consis2}
    \lim \limits_{n \to \infty} P(\hat{c} < c_0-\epsilon) = 0. 
\end{eqnarray}
We focus on proving \eqref{consis1}. The other part in \eqref{consis2} can be similarly proved. 
We choose  $\epsilon^{\ast} < \epsilon$ such that 
\begin{itemize}
    \item[](a) $H(x) \geq \frac{J_0+J_\epsilon}{2}$, for $x \in [c_0-\epsilon^{\ast},c_0+\epsilon^{\ast}]$;
    \item[](b) $\boldsymbol{\theta}_0^T\mathbf{Q}(c_0-\epsilon^{\ast}) < 0$ and $\boldsymbol{\theta}_0^T\mathbf{Q}(c_0+\epsilon^{\ast}) > 0$.
 \end{itemize}
By Conditions C1 and C2, the existence of such $\epsilon^{\ast}$ is obvious. 
We further define a subset of the sample space as 
$A_{n,\epsilon} = A_{n1,\epsilon} \cap A_{n2,\epsilon} \cap A_{n3,\epsilon},$
where
\begin{eqnarray*}
   A_{n1,\epsilon}&=&\left\{ \hat{\boldsymbol{\theta}}^T\mathbf{Q}(c_0-\epsilon^{\ast}) < \frac{1}{2}\boldsymbol{\theta}_0^T\mathbf{Q}(c_0-\epsilon^{\ast})\right\},  \\
   A_{n2,\epsilon}&=&\left\{ \hat{\boldsymbol{\theta}}^T\mathbf{Q}(c_0+\epsilon^{\ast}) > \frac{1}{2}\boldsymbol{\theta}_0^T\mathbf{Q}(c_0+\epsilon^{\ast})\right\}, \\
   A_{n3,\epsilon}& =&\left\{  \inf_{x \in [c_0-\epsilon^{\ast},c_0+\epsilon^{\ast}]} \hat{H}(x) \geq \frac{J_0+3J_{\epsilon}}{4}\right\}. 
\end{eqnarray*}
The two subsets $A_{n1,\epsilon}$ and $A_{n2,\epsilon}$ together ensure that there exists a solution $\hat c^*$  to $\hat{\boldsymbol{\theta}}^T\mathbf{Q}(x) = 0$ in $[c_0-\epsilon^{\ast},c_0+\epsilon^{\ast}]$,  and $A_{n3,\epsilon}$ implies that  $\hat{H}(\hat c^*)$ is very close to $J_0$.

With the choice of $\epsilon^{\ast}$, the consistency of $\hat{\boldsymbol{\theta}}$ \citep{cai2018empirical}, 
and the fact that $\Delta_n=O_p(n^{-1/2})$, it can be shown that
\begin{equation}
\label{limit.An}
\lim \limits_{n \to \infty} P(A_{n1,\epsilon}) = \lim \limits_{n \to \infty} P(A_{n2,\epsilon}) =
\lim \limits_{n \to \infty} P(A_{n3,\epsilon}) = 1. 
\end{equation}
The details are sketched as follows. By the choice of $\epsilon^{\ast}$, 
\begin{eqnarray*}
    P(A_{n1,\epsilon}) &=& P\left(\hat{\boldsymbol{\theta}}^T\mathbf{Q}(c_0-\epsilon^{\ast}) -  \boldsymbol{\theta}_0^T\mathbf{Q}(c_0-\epsilon^{\ast})< -\frac{1}{2}\boldsymbol{\theta}_0^T\mathbf{Q}(c_0-\epsilon^{\ast})\right)\\
    &\geq& P\left(\big|\hat{\boldsymbol{\theta}}^T\mathbf{Q}(c_0-\epsilon^{\ast}) -\boldsymbol{\theta}_0^T\mathbf{Q}(c_0-\epsilon^{\ast})\big|< -\frac{1}{2}\boldsymbol{\theta}_0^T\mathbf{Q}(c_0-\epsilon^{\ast})\right). 
\end{eqnarray*}
Then by the consistency of $\hat{\boldsymbol{\theta}}$ \citep{cai2018empirical}, 
we have 
$\lim \limits_{n \to \infty} P(A_{n1,\epsilon}) = 1$. 
Similarly, we also have 
$
\lim \limits_{n \to \infty} P(A_{n2,\epsilon}) = 1$.
As for the third term $A_{n3,\epsilon}$, again by the choice of $\epsilon^{\ast}$, 
when $x \in [c_0-\epsilon^{\ast},c_0+\epsilon^{\ast}]$, we have 
\begin{eqnarray*}
\hat{H}(x) &=& \{\hat{H}(x) - H(x) + H(x)\}
    \geq -\Delta_{n} + \frac{J_0+J_{\epsilon}}{2}.
\end{eqnarray*}
Therefore, 
$$
P(A_{n3,\epsilon}) \geq P\left(
-\Delta_{n} + \frac{J_0+J_{\epsilon}}{2}\geq \frac{J_0+3J_{\epsilon}}{4}
\right)= P\left(\Delta_{n} \leq \frac{J_0-J_{\epsilon}}{4}\right).
$$
Since $\Delta_{n} = O_p(n^{-\frac{1}{2}})$, we have 
$
\lim \limits_{n \to \infty} P(A_{n3,\epsilon}) = 1$. 

We are now ready to prove \eqref{consis1}.  Note that 
\begin{eqnarray*}
P(\hat{c} > c_0+\epsilon)&\leq &P\big(
H(\hat{c})\leq J_{\epsilon}
\big)\leq P\big(
\hat H(\hat{c})\leq J_{\epsilon}+\Delta_n
\big)\\
&\leq&P\Big(\{\hat H(\hat{c})\leq J_{\epsilon}+\Delta_n \}\cap A_{n,\epsilon}\Big)+P(A_{n,\epsilon}^c).
\end{eqnarray*}
By the definition of $A_{n,\epsilon}$, if $\{\hat H(\hat{c})\leq J_{\epsilon}+\Delta_n \}$ and $A_{n,\epsilon}$ both occur, we have 
$$
J_{\epsilon}+\Delta_n \geq \hat H(\hat{c})\geq \hat H(\hat{c}^*)\geq \inf_{x \in [c_0-\epsilon^{\ast},c_0+\epsilon^{\ast}]} \hat{H}(x) \geq \frac{J_0+3J_{\epsilon}}{4}, 
$$
which implies 
$\Delta_{n} \geq (J_0-J_{\epsilon})/{4}.$
Hence, 
 \begin{eqnarray*}
P(\hat{c} > c_0+\epsilon)&\leq &P\Big(\Delta_{n} \geq \frac{J_0-J_{\epsilon}}{4} \Big)+P(A_{n,\epsilon}^c)
\to 0,
\end{eqnarray*}
where the last step follows from  \eqref{limit.An} and $\Delta_{n} = O_p(n^{-\frac{1}{2}})$. 
This finishes the proof of \eqref{consis1} and the consistency of $\hat c$ stated in \eqref{lemma.result1}. 

For \eqref{lemma.result2}, we note that  
$$
A_{n1,\epsilon}\cap A_{n2,\epsilon}
\subset 
\left\{   \hat{\boldsymbol{\theta}}^T\mathbf{Q}(\hat c) = 0\right\},
$$
which, together with (\ref{limit.An}), implies that 
$$
\lim \limits_{n \to \infty}  P\left(   \hat{\boldsymbol{\theta}}^T\mathbf{Q}(\hat c) = 0\right)=1.
$$
This completes the proof of (\ref{lemma.result2}). 

\end{proof}

\subsection*{A.3 \ Proof of Theorem 1}

We first consider Part (a). 
By (\ref{lemma.result2}) of Lemma 1 and the Slutsky's theorem, 
we can derive the asymptotic normality of $\hat c$ from 
$\hat{\boldsymbol{\theta}}^T\mathbf{Q}(\hat c) = 0$. 
Applying the first-order Taylor expansion on $\mathbf{q}(\hat{c})$ at the point $x=c_0$ and using 
the consistency result of $\hat c$ in \eqref{lemma.result1} of Lemma 1, 
we have 
\begin{equation*}
  0 = \hat{\alpha} + \hat{\boldsymbol{\beta}}^T\mathbf{q}(c_0) + \hat{\boldsymbol{\beta}}^T \dot{\mathbf{q}}(c_0)(\hat{c} - c_0) + o_p(1)\cdot (\hat{c} - c_0).  
\end{equation*}
By Theorem 1 of \cite{cai2018empirical}, we have 
\begin{equation}
\sqrt{n}
(\hat{\boldsymbol{\theta}} - \boldsymbol{\theta}_0) \to {\rm N} (0, \S^{-1}\V\S^{-1})
\end{equation}
 in distribution as $n \to \infty$. This together with the fact $\boldsymbol{\theta}_0^T\mathbf{Q}(c_0) = 0$  implies that 
 \begin{equation*}
  \sqrt{n}(\hat{c}-c_0) =- \frac{\mathbf{Q}^T(c_0) }{\boldsymbol{\beta}_0^T\dot{\mathbf{q}}(c_0)}
  \left\{ \sqrt{n}
(\hat{\boldsymbol{\theta}}- \boldsymbol{\theta}_0)\right\}+ o_p(1) \to {\rm N} (0,\sigma_{c}^2)
\end{equation*}
in distribution as $n \to \infty$,  
where $\sigma_{c}^2$ is defined in \eqref{sigma_cr}.

We next consider Part (b). 
Recall that 
\begin{eqnarray*}
\hat{J} - J_0&=&  \{\hat{F}_{0}(\hat{c}) -\hat{F}_{1}(\hat{c})\} - \{F_0(c_0) - F_1(c_0)\}.
\end{eqnarray*}
Let 
\begin{eqnarray*}
M_{n0}&=&\hat{F}_{0}(c_0) -F_0(c_0),~~
M_{n1}= \hat{F}_{1}(c_0) - F_1(c_0),\\ 
{e}_{n0}&=&\{\hat{F}_{0}(\hat{c}) -\hat{F}_{0}(c)\} -\{F_0(\hat{c}) - F_0(c_0)\},\\
{e}_{n1}&=&\{\hat{F}_{1}(\hat{c}) -\hat{F}_{1}(c)\} -\{F_1(\hat{c}) - F_1(c_0)\},\\
{e}_{n2}&=&\{F_0(\hat{c}) - F_1(\hat{c})\}-\{F_0(c_0) - F_1(c_0)\}. 
\end{eqnarray*}
It can be shown that  
\begin{eqnarray}
\hat{J} - J_0=M_{n0}-M_{n1}+{e}_{n0}+{e}_{n1}+{e}_{n2}. 
\end{eqnarray}
One of the key technical arguments is to show that ${e}_{n0}$, ${e}_{n1}$, and ${e}_{n2}$ are all of order $o_p(n^{-1/2})$. 

By Lemma 4 of \cite{cai2018empirical}, we have 
 for any $b>0$, 
   \begin{eqnarray}
       && \sup \limits_{x:|x-c_0|<bn^{-1/2}}|\{\hat{F}_{0}(x) -\hat{F}_{0}(c_0)\} -\{F_0(x) - F_0(c_0)\}| \nonumber \\
        &=& O_p(n^{-3/4}(\log(n))^{1/2})=o_p(n^{-1/2}).
        \label{result.sup.order}
    \end{eqnarray}
The result in  Part (a) implies that $\hat c-c_0=O_p(n^{-1/2})$, 
which, together with \eqref{result.sup.order}, leads to $e_{n0}=o_p(n^{-1/2})$. 
Similarly, we also have $e_{n1}=o_p(n^{-1/2})$. 
By the second order Taylor expansion and Condition A2, 
we have $e_{n2}=o_p(n^{-1/2})$.
It follows that 
\begin{eqnarray}
\sqrt{n}\left(\hat{J} - J_0\right)=\sqrt{n}(M_{n0}-M_{n1})+o_p(1). 
\end{eqnarray}
Applying Theorem 2 of \cite{cai2018empirical}, we have 
\begin{equation}
   \sqrt{n} \left( \begin{array}{c} M_{n0}\\ M_{n1}\end{array} \right) = \sqrt{n} \left( \begin{array}{c} \hat{F}_{0}(c_0) -F_0(c_0)\\ \hat{F}_{1}(c_0) -F_1(c_0)\end{array} \right) \to {\rm N} \left( \mathbf{0}, \left( \begin{array}{cc}
\sigma^2_{00} &\sigma^2_{01} \\
\sigma^2_{01} & \sigma^2_{11}
\end{array} \right) \right)
\end{equation}
in distribution as $n\to\infty$,  where 
\begin{eqnarray*}
\sigma^2_{00} &=& (1+\rho)\{F_0(c_0) - F_0^2(c_0)\} \\
                        & &  \;\;\;\;\;\;\;\;\; - \rho(1+\rho)\left\{ A_{0}(c_0) - \left( \begin{array}{c}
A_{0}(c_0)\\ \A_{1}(c_0)\end{array} \right)^T \A^{-1} \left(\begin{array}{c}
A_{0}(c_0)\\ \A_{1}(c_0)\end{array}\right) \right\},\\
  \sigma^2_{01}&=&(1+\rho)\left\{ A_{0}(c_0) - \left( \begin{array}{c}
A_{0}(c_0)\\ \A_{1}(c_0)\end{array} \right)^T \A^{-1} \left(\begin{array}{c}
A_{0}(c_0)\\ \A_{1}(c_0)\end{array}\right) \right\},\\
    \sigma^2_{11} &=& \frac{1+\rho}{\rho}\{F_1(c_0) - F_1^2(c)\} \\
                            & & \;\;\;\;\;\;\;\;\;\;\; - \frac{1+\rho}{\rho}\left\{ A_{0}(c_0) - \left( \begin{array}{c}
A_{0}(c_0)\\ \A_{1}(c_0)\end{array} \right)^T \A^{-1} \left(\begin{array}{c}
A_{0}(c_0)\\ \A_{1}(c_0)\end{array} \right) \right\}.
\end{eqnarray*}

It immediately follows that,  as $n\to\infty$, 
$$\sqrt{n}\left(M_{n0}-M_{n1}\right) \to {\rm N} (0, \sigma_{J}^2)$$
in distribution, where $\sigma_{J}^2$ is defined in \eqref{sigma_Jr}.
Recall that  $\sqrt{n} (\hat{J}_r - J_0) = \sqrt{n}\left(M_{n0}-M_{n1}\right)+o_p(1)$. 
By the Slusky's theorem,  we have
$$\sqrt{n} (\hat{J} - J_0) \to {\rm N} (0, \sigma_{J}^2)$$
in distribution as $n\to\infty$. This completes the proof of the theorem.


%

\newpage

{\centering {\large {\bf Supplementary Material for \\
``Semiparametric Inference of the Youden Index and the Optimal Cutoff Point under  Density Ratio Models"}} \par}

%

\bigskip

\no
Section S1 presents more details about the simulation settings.
Section S2 presents additional results for the gamma distributional setting in Section 3.1 of the main paper with the lower limit of detection (LLOD) equal to the 30\% quantile of the population $F_0$. 
Section S3 presents the  results for the lognormal distributional setting.

\section*{S1. \ Simulation Settings}
Recall that in the main paper, we consider the following two distributional settings in the simulation studies: 
\begin{itemize}
    \item [] (1) $f_0  \sim {\rm Gamma}(2,0.5)~{\rm and}~f_1  \sim {\rm Gamma}(2,\eta);$
    \item [] (2) $f_0  \sim {\rm LN}(2.5,0.09)~{\rm and} ~f_1  \sim {\rm LN}(\mu,0.25).$
\end{itemize}
For each distributional setting, we choose four values of $\eta$ or $\mu$  such that the corresponding 
Youden indexes equal  0.2, 0.4, 0.6, and 0.8. 
Table \ref{Settings} provides the exact values of $\eta$ or $\mu$ and the optimal cutoff points $c$ along with 
$F_0(c)$ and $F_1(c)$.

\begin{table}[!htt]
    \centering
    \small
    \caption{Parameter values in simulation studies}
    \label{Settings}
    \begin{tabular}{cccccc}
    \hline
     Distribution & $J$   &  $\eta/\mu$& c & $F_0(c)$& $F_1(c)$\\
     \hline
    \multirow{4}{*}{Gamma}& 0.20&0.34&4.79&0.69&0.49\\
    &0.40&0.23& 5.75&0.78&0.38\\
    & 0.60&0.14& 7.02&0.86&0.26\\
    & 0.80& 0.07& 9.04&0.94&0.14\\\hline
    \multirow{4}{*}{Lognormal}&0.2&2.62&16.92&0.86&0.66\\
     &0.40&2.87& 16.54&0.85&0.45\\
     &0.60&3.14& 17.30&0.88&0.28\\
     &0.80& 3.501&19.12&0.93&0.13\\
     \hline
    \end{tabular}
\end{table}

For each scenario, in addition to the case that there is no LLOD, we also consider 
two cases with fixed and finite LLODs, which are equal to 15\% quantile of $F_0$ and 
30\% quantile of $F_0$, respectively. 
The exact values of the LLODs are given in Table \ref{valLLOD}. 

\begin{table}[!htt]
    \centering
    \small
    \caption{The exact values of finite and finite LLODs}
    \label{valLLOD}
    \begin{tabular}{ccc}
    \hline
     Distribution &15\% quantile of $F_0$& 30\% quantile of $F_0$\\
     \hline
    Gamma &1.37&2.19\\
    Lognormal & 8.93 & 10.41\\
    \hline
    \end{tabular}
\end{table}

\section*{S2. \ Additional Simulation For Gamma Distributional Setting}

Tables~\ref{point.J.lod3}--\ref{point.c.lod3} compare the relative biases (RBs) and mean squared errors (MSEs) of point estimators of $(J,c)$ under gamma setting when the LLOD equals 30\% quantile of $F_0$.
Tables~\ref{CI.J.lod3}--\ref{CI.c.lod3} present the coverage probabilities (CPs) and average lengths (ALs) of the confidence intervals (CIs) of $(J,c)$ under the same setting. 
The general trend for comparing our proposed method and all candidate methods is similar to the case when the LLOD is equal to the 15\% quantile of $F_0$. Hence, we omit the comparison results here. 

\begin{table}[!htt]
\centering
\small
\caption{RB (\%) and MSE ($\times 100$) for point estimators of $J$ when the LLOD equals 30\% quantile of $F_0$} \label{point.J.lod3}
\begin{tabular}{cccccccccccc}
  \hline
  &$(n_0, n_1)$& \multicolumn{2}{c}{$(50, 50)$} & \multicolumn{2}{c}{$(100, 100)$} & \multicolumn{2}{c}{$(200, 200)$}& \multicolumn{2}{c}{$(50, 150)$} & \multicolumn{2}{c}{$(150, 50)$}\\
  \hline
  $J$& & RB & MSE & RB & MSE& RB & MSE& RB & MSE& RB & MSE\\ \hline
\multirow{4}{*}{$0.2$}& $\hat{J}$ & 8.62 & 0.64 & 3.78 & 0.31 & 1.71 & 0.14 & 6.89 & 0.41 & 4.69 & 0.41 \\ 
  &$\hat{J}_{B}$ & 10.52 & 0.65 & 5.12 & 0.32 & 2.94 & 0.16 & 8.24 & 0.42 & 6.35 & 0.43 \\ 
  &$\hat{J}_{G}$ & 5.85 & 0.80 & 0.77 & 0.40 & -2.84 & 0.20 & 4.31 & 0.53 & 0.83 & 0.57 \\ 
  &$\hat{J}_{E}$ & 39.99 & 1.29 & 26.6 & 0.64 & 17.72 & 0.30 & 32.97 & 0.86 & 32.52 & 0.89 \\ \hline
  \multirow{4}{*}{$0.4$}& $\hat{J}$ & 3.16 & 0.58 & 1.36 & 0.30 & 0.62 & 0.13 & 2.25 & 0.37 & 1.55 & 0.40 \\ 
  &$\hat{J}_{B}$ & 5.06 & 0.59 & 2.57 & 0.29 & 1.67 & 0.14 & 3.74 & 0.37 & 3.01 & 0.41 \\ 
  &$\hat{J}_{G}$ & -0.34 & 0.73 & -2.06 & 0.38 & -3.18 & 0.19 & -1.13 & 0.47 & -2.07 & 0.53 \\ 
  &$\hat{J}_{E}$ & 16.42 & 1.06 & 10.62 & 0.53 & 6.78 & 0.24 & 13.36 & 0.68 & 12.57 & 0.74 \\ \hline
  \multirow{4}{*}{$0.6$}& $\hat{J}$ & 1.81 & 0.46 & 0.66 & 0.24 & 0.33 & 0.11 & 1.11 & 0.28 & 0.89 & 0.35 \\ 
  &$\hat{J}_{B}$ & 3.14 & 0.44 & 1.64 & 0.23 & 1.11 & 0.11 & 2.16 & 0.27 & 2.03 & 0.33 \\ 
  &$\hat{J}_{G}$ & -1.12 & 0.56 & -1.91 & 0.30 & -2.42 & 0.15 & -1.72 & 0.37 & -1.90 & 0.42 \\ 
  &$\hat{J}_{E}$ & 8.89 & 0.76 & 5.38 & 0.37 & 3.70 & 0.19 & 7.14 & 0.48 & 6.73 & 0.56 \\ \hline
  \multirow{4}{*}{$0.8$}& $\hat{J}$ & 1.15 & 0.27 & 0.42 & 0.14 & 0.17 & 0.07 & 0.62 & 0.14 & 0.64 & 0.21 \\ 
  &$\hat{J}_{B}$ & 1.55 & 0.22 & 0.78 & 0.12 & 0.53 & 0.05 & 0.98 & 0.13 & 1.00 & 0.17 \\ 
  &$\hat{J}_{G}$ & -1.17 & 0.36 & -1.27 & 0.18 & -1.50 & 0.10 & -1.46 & 0.22 & -1.21 & 0.27 \\ 
  &$\hat{J}_{E}$ & 4.51 & 0.39 & 2.86 & 0.20 & 1.94 & 0.10 & 3.73 & 0.24 & 3.52 & 0.30 \\ \hline
  \end{tabular}
\end{table}
  
\begin{table}[!htt]
\centering
\small
\caption{RB (\%) and MSE ($\times100$) for point estimators of $c$ when the LLOD equals 30\% quantile of $F_0$}
\label{point.c.lod3}
\begin{tabular}{cccccccccccc}
  \hline
 &$(n_0, n_1)$& \multicolumn{2}{c}{$(50, 50)$} & \multicolumn{2}{c}{$(100, 100)$} & \multicolumn{2}{c}{$(200, 200)$}& \multicolumn{2}{c}{$(50, 150)$} & \multicolumn{2}{c}{$(150, 50)$}\\
   \hline
  $J$& & RB & MSE & RB & MSE& RB & MSE& RB & MSE& RB & MSE\\ \hline
 \multirow{4}{*}{$0.2$}& $\hat{c}$ & 1.72 & 166.59 & -0.14 & 76.36 & -1.11 & 25.61 & -0.17 & 100.80 & 0.16 & 99.95 \\ 
  &$\hat{c}_{B}$ & 1.22 & 120.35 & 1.62 & 64.03 & 0.53 & 30.39 & 1.06 & 84.80 & 1.17 & 83.72 \\ 
  &$\hat{c}_{G}$ & 18.38 & 241.37 & 20.51 & 170.02 & 19.84 & 126.64 & 19.37 & 187.87 & 22.02 & 267.68 \\ 
  &$\hat{c}_{E}$  & 0.21 & 246.68 & 0.65 & 169.08 & -0.76 & 116.80 & -0.59 & 204.52 & 2.51 & 215.93 \\ \hline
 \multirow{4}{*}{$0.4$}& $\hat{c}$ & -1.39 & 40.03 & -0.42 & 18.20 & -0.11 & 7.98 & -1.23 & 25.73 & -0.24 & 25.09 \\ 
  &$\hat{c}_{B}$ & 0.14 & 56.14 & 0.33 & 29.60 & 0.10 & 13.84 & -0.47 & 38.37 & 0.81 & 35.86 \\ 
  &$\hat{c}_{G}$ & 11.39 & 121.93 & 14.37 & 109.64 & 14.58 & 91.16 & 13.27 & 108.65 & 14.72 & 115.49 \\ 
  &$\hat{c}_{E}$  & -2.89 & 160.26 & -0.23 & 126.83 & -0.45 & 75.26 & -2.65 & 149.91 & 1.53 & 156.36 \\ \hline
 \multirow{4}{*}{$0.6$}& $\hat{c}$ & -0.59 & 35.72 & -0.27 & 17.64 & -0.04 & 8.36 & -0.81 & 24.8 & -0.01 & 22.01 \\ 
  &$\hat{c}_{B}$ & -0.32 & 51.67 & -0.27 & 26.14 & -0.20 & 12.4 & -1.09 & 33.34 & 0.40 & 30.83 \\ 
  &$\hat{c}_{G}$ & 6.94 & 123.88 & 9.08 & 90.88 & 10.55 & 78.96 & 8.81 & 93.44 & 9.86 & 98.53 \\ 
  &$\hat{c}_{E}$  & -2.29 & 161.37 & -0.67 & 118.08 & 0.20 & 74.7 & -2.40 & 147.73 & 1.22 & 146.64 \\ \hline
  \multirow{4}{*}{$0.8$}& $\hat{c}$ & -0.52 & 72.34 & -0.34 & 36.15 & 0.01 & 17.71 & -0.95 & 53.74 & 0.17 & 41.95 \\ 
  &$\hat{c}_{B}$ & -0.33 & 69.75 & -0.29 & 34.76 & -0.04 & 16.50 & -1.12 & 49.95 & 0.50 & 39.56 \\ 
  &$\hat{c}_{G}$ & 1.13 & 200.79 & 4.49 & 108.25 & 6.35 & 85.81 & 5.11 & 124.86 & 4.12 & 117.41 \\ 
  &$\hat{c}_{E}$  & -3.18 & 236.72 & -1.72 & 159.92 & -1.25 & 107.68 & -4.03 & 195.3 & 1.14 & 197.76 \\ \hline
  \end{tabular}
\end{table}

\begin{table}[!htt]
\centering
\small
\caption{CP (\%) and AL for CIs of $J$ when the LLOD equals 30\% quantile of $F_0$}
\label{CI.J.lod3}
\begin{tabular}{cccccccccccc}
  \hline
 &$(n_0, n_1)$& \multicolumn{2}{c}{$(50, 50)$} & \multicolumn{2}{c}{$(100, 100)$} & \multicolumn{2}{c}{$(200, 200)$}& \multicolumn{2}{c}{$(50, 150)$} & \multicolumn{2}{c}{$(150, 50)$}\\
  \hline
   $J$& & CP & AL & CP & AL& CP & AL& CP & AL& CP & AL\\\hline
  \multirow{4}{*}{$0.2$}& $\mathcal{I}_J$  & 92.1 & 0.31 & 93.6 & 0.21 & 96.0 & 0.15 & 92.8 & 0.24 & 93.7 & 0.25 \\ 
  &$\mathcal{I}_{JB}$ & 93.8 & 0.30 & 94.1 & 0.21 & 95.3 & 0.15 & 92.8 & 0.24 & 94.6 & 0.25 \\ 
  &$\mathcal{I}_{JG}$  & 93.2 & 0.33 & 95.2 & 0.24 & 95.3 & 0.17 & 93.6 & 0.27 & 94.0 & 0.27 \\ 
  &$\mathcal{I}_{JE}$ & 55.5 & 0.25 & 59.9 & 0.19 & 67.3 & 0.14 & 52.9 & 0.21 & 55.6 & 0.21 \\ \hline
  \multirow{4}{*}{$0.4$}& $\mathcal{I}_J$  & 95.0 & 0.28 & 95.1 & 0.20 & 95.8 & 0.15 & 93.8 & 0.23 & 94.28 & 0.24 \\ 
  &$\mathcal{I}_{JB}$ & 95.6 & 0.29 & 93.7 & 0.20 & 95.4 & 0.14 & 92.9 & 0.22 & 94.08 & 0.24 \\ 
  &$\mathcal{I}_{JG}$  & 93.8 & 0.33 & 94.3 & 0.23 & 94.9 & 0.16 & 94.3 & 0.26 & 93.18 & 0.27 \\ 
  &$\mathcal{I}_{JE}$ & 72.1 & 0.26 & 75.0 & 0.19 & 78.8 & 0.14 & 69.2 & 0.21 & 73.22 & 0.23 \\ \hline
  \multirow{4}{*}{$0.6$}& $\mathcal{I}_J$  & 95.6 & 0.26 & 94.5 & 0.18 & 95.28 & 0.13 & 94.3 & 0.20 & 93.79 & 0.22 \\ 
  &$\mathcal{I}_{JB}$ & 95.8 & 0.25 & 94.1 & 0.18 & 95.08 & 0.13 & 94.2 & 0.19 & 94.29 & 0.22 \\ 
  &$\mathcal{I}_{JG}$  & 94.8 & 0.29 & 94.8 & 0.20 & 94.57 & 0.14 & 94.1 & 0.22 & 93.29 & 0.25 \\ 
  &$\mathcal{I}_{JE}$ & 76.0 & 0.24 & 78.1 & 0.18 & 81.61 & 0.13 & 72.8 & 0.19 & 77.35 & 0.21 \\ \hline
  \multirow{4}{*}{$0.8$}& $\mathcal{I}_J$  & 95.5 & 0.20 & 94.6 & 0.14 & 95.18 & 0.10 & 95.3 & 0.14 & 95.48 & 0.18 \\ 
  &$\mathcal{I}_{JB}$ & 95.8 & 0.18 & 95.0 & 0.13 & 95.08 & 0.09 & 94.9 & 0.14 & 95.58 & 0.17 \\ 
  &$\mathcal{I}_{JG}$  & 94.4 & 0.23 & 93.8 & 0.16 & 94.18 & 0.11 & 94.1 & 0.16 & 94.18 & 0.20 \\ 
  &$\mathcal{I}_{JE}$ & 80.2 & 0.18 & 84.7 & 0.14 & 84.94 & 0.10 & 73.7 & 0.13 & 81.22 & 0.16 \\ 
  \hline
\end{tabular}
\end{table}

\begin{table}[!htt]
\centering
\small
\caption{CP (\%) and AL for CIs of $c$ when the LLOD equals 30\% quantile of $F_0$}
\label{CI.c.lod3}
\begin{tabular}{cccccccccccc}
  \hline
 &$(n_0, n_1)$& \multicolumn{2}{c}{$(50, 50)$} & \multicolumn{2}{c}{$(100, 100)$} & \multicolumn{2}{c}{$(200, 200)$}& \multicolumn{2}{c}{$(50, 150)$} & \multicolumn{2}{c}{$(150, 50)$}\\
  \hline
   $J$& & CP & AL & CP & AL& CP & AL& CP & AL& CP & AL\\\hline 
  \multirow{4}{*}{$0.2$}& $\mathcal{I}_c$& 96.6 & 4.37 & 95.3 & 3.27 & 96.4 & 2.02 & 97.7 & 4.67 & 94.7 & 4.42 \\ 
  &$\mathcal{I}_{cB}$ & 81.5 & 4.53 & 81.3 & 2.14 & 83.7 & 1.53 & 70.4 & 2.19 & 85.1 & 2.69 \\ 
   &$\mathcal{I}_{cG}$ & 89.8 & 5.34 & 75.7 & 3.59 & 63.6 & 2.42 & 80.9 & 4.42 & 80.6 & 4.17 \\ 
  &$\mathcal{I}_{cE}$  & 95.9 & 4.72 & 94.5 & 4.13 & 95.5 & 3.55 & 93.8 & 4.25 & 95.0 & 4.52 \\ \hline
  \multirow{4}{*}{$0.4$}& $\mathcal{I}_c$& 95.4 & 2.37 & 94.7 & 1.61 & 95.7 & 1.12 & 94.8 & 1.86 & 95.6 & 1.91 \\ 
  &$\mathcal{I}_{cB}$ & 84.4 & 2.19 & 84.0 & 1.57 & 86.8 & 1.13 & 75.4 & 1.47 & 89.1 & 1.94 \\ 
   &$\mathcal{I}_{cG}$ & 89.5 & 3.39 & 75.2 & 2.45 & 51.2 & 1.75 & 78.1 & 2.68 & 73.8 & 2.49 \\ 
  &$\mathcal{I}_{cE}$  & 94.1 & 4.27 & 95.4 & 3.66 & 95.1 & 3.09 & 93.0 & 3.78 & 96.4 & 4.18 \\ \hline
  \multirow{4}{*}{$0.6$}& $\mathcal{I}_c$& 92.4 & 2.28 & 94.3 & 1.62 & 95.1 & 1.14 & 92.9 & 1.90 & 94.8 & 1.83 \\ 
  &$\mathcal{I}_{cB}$ & 89.5 & 2.34 & 89.2 & 1.67 & 91.1 & 1.19 & 83.9 & 1.65 & 93.3 & 1.98 \\ 
   &$\mathcal{I}_{cG}$ & 93.7 & 3.67 & 86.4 & 2.68 & 69.5 & 1.92 & 87.1 & 2.80& 83.4 & 2.59 \\ 
  &$\mathcal{I}_{cE}$  & 94.4 & 4.31 & 95.1 & 3.68 & 96.4 & 3.03 & 91.8 & 3.82 & 94.6 & 4.14 \\ \hline
   \multirow{4}{*}{$0.8$}& $\mathcal{I}_c$& 93.0 & 3.25 & 94.0 & 2.31 & 95.0 & 1.65 & 93.0 & 2.78 & 94.28 & 2.51 \\ 
  &$\mathcal{I}_{cB}$ & 92.4 & 3.01 & 93.0 & 2.14 & 93.47 & 1.52 & 89.9 & 2.43 & 93.88 & 2.35 \\ 
   &$\mathcal{I}_{cG}$ & 94.9 & 5.43 & 94.9 & 3.86 & 89.46 & 2.75 & 94.2 & 3.90 & 94.58 & 3.44 \\ 
  &$\mathcal{I}_{cE}$  & 88.5 & 4.88 & 93.7 & 4.24 & 94.28 & 3.53 & 84.7 & 4.34 & 94.48 & 4.73 \\\hline 
\end{tabular}
\end{table}

\clearpage
\section*{S3. \ Additional Simulation For Lognormal Distributional Setting}
In this section, we present the simulation results under the lognormal distributional setting. 
Tables~\ref{point.J.log}--\ref{CI.c.log} provide the simulation results of the point estimators and CIs of $(J,c)$  when there is no LLOD.
Tables~\ref{point.J.log.lod}--\ref{CI.c.log.lod} summarize the simulation results of the point estimators and CIs of $(J,c)$ when the LLOD equals 15\% quantile of $F_0$.
Tables~\ref{point.J.log.lod3}--\ref{CI.c.log.lod3} summarize the simulation results of the point estimators and CIs of $(J,c)$ when the LLOD equals 30\% quantile of $F_0$.
We only summarize the comparison results between our proposed method and the Box-Cox method. 
The general trend for comparing our method and other candidate methods is similar to the gamma distributional setting. Hence
we omit their comparison. 

First, we discuss the point estimators of $(J,c)$. For estimating the Youden index, the RBs and MSEs of the estimators $\hat{J}$ and $\hat{J}_B$ are very close and small in majority cases. 
For estimating the optimal cutoff point, the estimator $\hat{c}_B$ is uniformly better than our estimator  in terms of MSE. 
This is expected because the parametric assumption for the Box-Cox method is satisfied.

Next, we discuss the findings for the CIs of $(J,c)$. In general, the ALs of both $\mathcal{I}_{J}$ and $\mathcal{I}_{JB}$ are comparable and small, while both CIs encounter slight overcoverage in some cases especially in the cases that one of the sample sizes is small. 
The performance of the CI $\mathcal{I}_{c}$ is stable with short ALs and reasonable CPs when there is no LLOD or when the LLOD equals 15\% quantile of $F_0$. When the LLOD increases to 30\% quantile of $F_0$, the CI $\mathcal{I}_{c}$ tends to have  undercoverage and longer AL especially in the cases when one of small sample sizes is small or when the Youden index is small. 
When there is no LLOD, the CI $\mathcal{I}_{cB}$ has similar performance as $\mathcal{I}_{c}$. 
However, with the existence of a fixed and finite LLOD,  the CI 
$\mathcal{I}_{cB}$ experiences severe undercoverage when $J = 0.2$ and $0.4$. 
Consequently, the CPs of $\mathcal{I}_{cB}$ are much worse than those of $\mathcal{I}_{c}$ in those cases.

\begin{table}[!htt]
\centering
\small
\caption{RB (\%) and MSE ($\times100$) for point estimators of $J$ when there is no LLOD \label{point.J.log}}
\begin{tabular}{cccccccccccc}
  \hline
 &$(n_0, n_1)$& \multicolumn{2}{c}{$(50, 50)$} & \multicolumn{2}{c}{$(100, 100)$} & \multicolumn{2}{c}{$(200, 200)$}& \multicolumn{2}{c}{$(50, 150)$} & \multicolumn{2}{c}{$(150, 50)$}\\
  \hline
  $J$& & RB & MSE & RB & MSE& RB & MSE& RB & MSE& RB & MSE\\ \hline
   \multirow{6}{*}{$0.2$}& $\hat{J}$ & 5.37 & 0.48 & 3.22 & 0.25 & 0.83 & 0.14 & 4.15 & 0.26 & 3.34 & 0.41 \\ 
   &$\hat{J}_{B}$ & 2.80 & 0.46 & 1.61 & 0.23 & 0.19 & 0.13 & 3.03 & 0.25 & 0.31 & 0.38 \\ 
   &$\hat{J}_{G}$ & 2.10 & 0.48 & 1.42 & 0.25 & -0.27 & 0.14 & 1.84 & 0.26 & 2.14 & 0.42 \\ 
   &$\hat{J}_{E}$ & 29.91 & 0.88 & 20.20 & 0.46 & 12.17 & 0.22 & 24.49 & 0.54 & 24.86 & 0.69 \\ 
   &$\hat{J}_{K}$ & 8.50 & 0.49 & 6.00 & 0.28 & 3.12 & 0.16 & 5.59 & 0.29 & 8.20 & 0.43 \\ 
   &$\hat{J}_{H}$ & 9.48 & 0.60 & 5.96 & 0.29 & 2.89 & 0.16 & 7.85 & 0.34 & 6.54 & 0.48 \\ \hline
  \multirow{6}{*}{$0.4$}& $\hat{J}$ & 2.87 & 0.55 & 1.55 & 0.26 & 0.27 & 0.15 & 1.99 & 0.29 & 1.36 & 0.44 \\ 
  &$\hat{J}_{B}$ & 2.60 & 0.52 & 1.38 & 0.25 & 0.23 & 0.14 & 1.90 & 0.28 & 0.91 & 0.40 \\ 
  &$\hat{J}_{G}$ & 1.06 & 0.53 & 0.54 & 0.25 & -0.27 & 0.15 & 0.55 & 0.28 & 0.64 & 0.44 \\ 
  &$\hat{J}_{E}$ & 14.10 & 0.91 & 9.31 & 0.43 & 5.54 & 0.23 & 11.46 & 0.54 & 10.99 & 0.67 \\ 
  &$\hat{J}_{K}$ & 2.08 & 0.49 & 1.18 & 0.25 & 0.17 & 0.16 & 1.03 & 0.29 & 1.72 & 0.40 \\ 
  &$\hat{J}_{H}$ & 2.12 & 0.64 & 2.01 & 0.31 & 0.97 & 0.18 & 3.04 & 0.35 & 1.29 & 0.50 \\ \hline
  \multirow{6}{*}{$0.6$}& $\hat{J}$ & 1.90 & 0.45 & 0.92 & 0.21 & 0.20 & 0.12 & 1.15 & 0.24 & 0.76 & 0.35 \\ 
  &$\hat{J}_{B}$ & 2.18 & 0.41 & 1.10 & 0.19 & 0.26 & 0.11 & 1.35 & 0.22 & 0.95 & 0.31 \\ 
  &$\hat{J}_{G}$ & 0.32 & 0.43 & 0.13 & 0.21 & -0.26 & 0.12 & -0.11 & 0.23 & 0.23 & 0.36 \\ 
  &$\hat{J}_{E}$ & 7.97 & 0.69 & 5.16 & 0.34 & 3.16 & 0.17 & 6.48 & 0.41 & 5.91 & 0.51 \\ 
  &$\hat{J}_{K}$ & -1.15 & 0.37 & -1.38 & 0.19 & -1.43 & 0.12 & -1.21 & 0.23 & -1.60 & 0.29 \\ 
  &$\hat{J}_{H}$ & 0.65 & 0.52 & 0.26 & 0.25 & -0.38 & 0.14 & 0.84 & 0.28 & 0.05 & 0.43 \\ \hline
  \multirow{6}{*}{$0.8$}& $\hat{J}$ & 1.42 & 0.25 & 0.64 & 0.12 & 0.18 & 0.07 & 0.72 & 0.13 & 0.63 & 0.19 \\ 
  &$\hat{J}_{B}$ & 1.36 & 0.21 & 0.66 & 0.10 & 0.19 & 0.05 & 0.78 & 0.11 & 0.61 & 0.16 \\ 
  &$\hat{J}_{G}$ & -0.12 & 0.26 & -0.24 & 0.13 & -0.31 & 0.07 & -0.59 & 0.14 & 0.03 & 0.20 \\ 
  &$\hat{J}_{E}$ & 4.61 & 0.39 & 2.91 & 0.19 & 1.86 & 0.10 & 3.70 & 0.23 & 3.20 & 0.28 \\ 
  &$\hat{J}_{K}$ & -3.27 & 0.28 & -3.02 & 0.16 & -2.61 & 0.10 & -2.56 & 0.17 & -3.77 & 0.25 \\ 
  &$\hat{J}_{H}$ & 1.38 & 0.38 & 1.00 & 0.19 & 0.35 & 0.10 & 1.04 & 0.19 & 0.92 & 0.31 \\ 
   \hline
  \end{tabular}
\end{table}

\begin{table}[!htt]
\centering
\small
\caption{RB (\%) and MSE ($\times100$) for point estimators of $c$ when there is no LLOD \label{point.c.log}}
\begin{tabular}{cccccccccccc}
  \hline
&$(n_0, n_1)$& \multicolumn{2}{c}{$(50, 50)$} & \multicolumn{2}{c}{$(100, 100)$} & \multicolumn{2}{c}{$(200, 200)$}& \multicolumn{2}{c}{$(50, 150)$} & \multicolumn{2}{c}{$(150, 50)$}\\
  \hline
  $J$& & RB & MSE & RB & MSE& RB & MSE& RB & MSE& RB & MSE\\ \hline
  \multirow{6}{*}{$0.2$}& $\hat{c}$ & 0.20 & 165.43 & 0.10 & 74.31 & 0.22 & 38.72 & -0.08 & 106.52 & 0.58 & 92.00 \\ 
  &$\hat{c}_{B}$ & -0.57 & 150.88 & -0.31 & 64.87 & 0.00 & 32.17 & -0.39 & 99.33 & -0.10 & 72.64 \\ 
  &$\hat{c}_{G}$ & -1.93 & 249.45 & -0.76 & 100.65 & -0.30 & 51.49 & -1.29 & 130.09 & -0.21 & 99.08 \\ 
  &$\hat{c}_{E}$ & -2.90 & 547.53 & -1.50 & 343.76 & -0.34 & 228.25 & -3.14 & 448.08 & -0.34 & 448.97 \\ 
  &$\hat{c}_{K}$ & 2.86 & 480.18 & 1.95 & 253.76 & 1.87 & 155.61 & 2.65 & 404.36 & 2.33 & 285.09 \\ 
  &$\hat{c}_{H}$ & 1.98 & 300.68 & 0.09 & 225.39 & 0.21 & 110.75 & -0.32 & 298.09 & 0.99 & 258.55 \\ \hline
 \multirow{6}{*}{$0.4$}& $\hat{c}$ & -0.26 & 86.21 & -0.05 & 41.36 & -0.02 & 20.08 & -0.16 & 56.25 & 0.14 & 53.16 \\ 
 &$\hat{c}_{B}$ & -0.85 & 77.38 & -0.38 & 38.58 & -0.18 & 18.08 & -0.43 & 53.77 & -0.42 & 42.94 \\ 
 &$\hat{c}_{G}$ & -2.17 & 159.00 & -1.10 & 71.71 & -0.54 & 36.52 & -1.19 & 80.69 & -0.51 & 72.77 \\ 
  &$\hat{c}_{E}$ & -1.58 & 347.41 & -0.57 & 242.58 & -0.44 & 151.00 & -1.55 & 285.24 & 0.25 & 299.00 \\ 
  &$\hat{c}_{K}$ & 3.33 & 211.66 & 2.62 & 129.96 & 2.10 & 75.34 & 2.67 & 187.62 & 3.06 & 130.84 \\ 
  &$\hat{c}_{H}$ & 0.79 & 279.02 & 0.60 & 141.91 & 0.47 & 66.56 & 0.38 & 181.29 & 1.00 & 178.57 \\ \hline
  \multirow{6}{*}{$ 0.6$}& $\hat{c}$ & -0.66 & 76.57 & -0.29 & 36.89 & -0.20 & 17.93 & -0.43 & 48.70 & -0.24 & 47.57 \\ 
  &$\hat{c}_{B}$ & -0.56 & 67.27 & -0.25 & 32.86 & -0.19 & 16.09 & -0.38 & 45.19 & -0.22 & 36.87 \\ 
  &$\hat{c}_{G}$ & -2.46 & 174.05 & -1.07 & 82.60 & -0.60 & 37.54 & -1.15 & 80.66 & -0.82 & 77.85 \\ 
  &$\hat{c}_{E}$ & -1.20 & 282.61 & -0.63 & 183.82 & -0.09 & 115.00 & -1.50 & 237.71 & 0.39 & 230.26 \\ 
  &$\hat{c}_{K}$ & 4.32 & 200.28 & 3.44 & 112.79 & 2.63 & 66.22 & 3.32 & 154.10 & 3.89 & 116.61 \\ 
  &$\hat{c}_{H}$ & 0.57 & 244.48 & 1.43 & 132.73 & 0.97 & 57.63 & 1.00 & 155.92 & 1.20 & 181.91 \\ \hline
  \multirow{6}{*}{$0.8$}& $\hat{c}$ & -0.88 & 99.18 & -0.56 & 48.62 & -0.41 & 24.13 & -0.84 & 65.16 & -0.48 & 60.61 \\ 
  &$\hat{c}_{B}$ & -0.22 & 70.70 & -0.13 & 34.19 & -0.16 & 16.96 & -0.36 & 51.67 & 0.02 & 37.30 \\ 
  &$\hat{c}_{G}$ & -3.48 & 287.68 & -1.58 & 137.17 & -0.77 & 63.57 & -1.01 & 116.65 & -1.94 & 134.35 \\ 
  &$\hat{c}_{E}$ & -1.37 & 291.59 & -0.55 & 189.65 & -0.74 & 128.34 & -1.46 & 232.99 & 0.43 & 233.55 \\ 
  &$\hat{c}_{K}$ & 3.66 & 199.67 & 3.16 & 114.55 & 2.37 & 65.77 & 3.04 & 163.09 & 3.33 & 114.45 \\ 
  &$\hat{c}_{H}$ & -1.78 & 247.56 & -1.45 & 135.38 & -1.38 & 78.60 & -1.95 & 156.22 & -1.49 & 194.25 \\ 
  \hline
\end{tabular}
\end{table}

\begin{table}[!htt]
\centering
\small
\caption{CP (\%) and AL for CIs of $J$ when there is no LLOD  \label{CI.J.log}}
\begin{tabular}{cccccccccccc}
  \hline
 &$(n_0, n_1)$& \multicolumn{2}{c}{$(50, 50)$} & \multicolumn{2}{c}{$(100, 100)$} & \multicolumn{2}{c}{$(200, 200)$}& \multicolumn{2}{c}{$(50, 150)$} & \multicolumn{2}{c}{$(150, 50)$}\\ 
  \hline
   $J$& & CP & AL & CP & AL& CP & AL& CP & AL& CP & AL\\\hline
 \multirow{5}{*}{$0.2$}& $\mathcal{I}_J$& 96.0 & 0.28 & 95.2 & 0.20 & 94.2 & 0.14 & 94.8 & 0.20 & 94.4 & 0.25 \\ 
  &$\mathcal{I}_{JB}$ & 95.2 & 0.27 & 94.0 & 0.19 & 93.8 & 0.13 & 95.8 & 0.20 & 93.8 & 0.23 \\ 
  &$\mathcal{I}_{JG}$ & 95.4 & 0.27 & 93.8 & 0.20 & 92.8 & 0.14 & 95.0 & 0.20 & 94.3 & 0.25 \\ 
  &$\mathcal{I}_{JE}$ & 80.6 & 0.29 & 82.7 & 0.21 & 82.8 & 0.15 & 73.5 & 0.21 & 83.7 & 0.26 \\ 
  &$\mathcal{I}_{JK}$ & 94.9 & 0.27 & 92.7 & 0.20 & 92.5 & 0.14 & 93.3 & 0.20 & 93.4 & 0.24 \\  
  &$\mathcal{I}_{JH}$ & 93.4 & 0.28 & 93 & 0.20 & 94.8 & 0.15 & 90 & 0.21 & 94.8 & 0.25 \\  \hline
  \multirow{5}{*}{$0.4$}& $\mathcal{I}_J$ & 96.2 & 0.28 & 96.0 & 0.20 & 94.0 & 0.15 & 95.2 & 0.21 & 95.2 & 0.25 \\ 
  &$\mathcal{I}_{JB}$ & 95.4 & 0.29 & 95.6 & 0.20 & 93.5 & 0.14 & 95.7 & 0.21 & 95.2 & 0.25 \\ 
  &$\mathcal{I}_{JG}$ & 95.9 & 0.29 & 94.6 & 0.20 & 93.9 & 0.14 & 95.3 & 0.21 & 93.9 & 0.26 \\ 
  &$\mathcal{I}_{JE}$ & 82.3 & 0.29 & 86.0 & 0.22 & 85.4 & 0.16 & 77.7 & 0.22 & 85.7 & 0.27 \\ 
  &$\mathcal{I}_{JK}$ & 94.8 & 0.28 & 94.7 & 0.20 & 92.7 & 0.14 & 94.3 & 0.21 & 93.7 & 0.24 \\
  &$\mathcal{I}_{JH}$ & 94.0 & 0.30 & 94.9 & 0.22 & 95.6 & 0.16 & 95.2 & 5.66 & 94.4 & 0.27 \\ \hline
  \multirow{5}{*}{$0.6$}& $\mathcal{I}_J$ & 95.7 & 0.26 & 96.2 & 0.18 & 94.8 & 0.13 & 95.8 & 0.19 & 95.3 & 0.23 \\ 
  &$\mathcal{I}_{JB}$ & 95.6 & 0.25 & 96.0 & 0.18 & 94.7 & 0.13 & 95.7 & 0.19 & 95.4 & 0.22 \\ 
  &$\mathcal{I}_{JG}$ & 94.3 & 0.26 & 94.9 & 0.18 & 93.7 & 0.13 & 94.4 & 0.19 & 93.9 & 0.23 \\ 
  &$\mathcal{I}_{JE}$ & 84.2 & 0.26 & 83.7 & 0.19 & 85.7 & 0.14 & 77.7 & 0.19 & 84.5 & 0.24 \\ 
  &$\mathcal{I}_{JK}$ & 94.3 & 0.24 & 95.8 & 0.17 & 93.1 & 0.13 & 95.1 & 0.19 & 94.4 & 0.21 \\
  &$\mathcal{I}_{JH}$ & 93.9 & 0.29 & 94.2 & 0.20 & 95.0& 0.14 & 92.1 & 0.20 & 94.7 & 0.27 \\\hline
  \multirow{5}{*}{$0.8$}& $\mathcal{I}_J$ & 96.9 & 0.20 & 95.7 & 0.14 & 94.9 & 0.10 & 96.0 & 0.15 & 94.9 & 0.17 \\ 
  &$\mathcal{I}_{JB}$ & 96.0 & 0.18 & 96.2 & 0.13 & 94.9 & 0.09 & 95.2 & 0.14 & 95.6 & 0.16 \\ 
  &$\mathcal{I}_{JG}$ & 92.8 & 0.20 & 94.5 & 0.14 & 94.6 & 0.10 & 94.0 & 0.14 & 92.5 & 0.18 \\ 
  &$\mathcal{I}_{JE}$ & 82.3 & 0.18 & 85.6 & 0.14 & 86 & 0.10 & 77.7 & 0.14 & 84.5 & 0.17 \\ 
  &$\mathcal{I}_{JK}$ & 94.2 & 0.18 & 92.5 & 0.13 & 88.9 & 0.09 & 94.1 & 0.14 & 91.6 & 0.16 \\ 
  &$\mathcal{I}_{JH}$ & 89.5 & 0.23 & 95.7 & 0.17 & 94.7 & 0.12 & 92.4 & 0.17 & 91.3 & 0.22 \\
  \hline
   \end{tabular}
\end{table}

\begin{table}[!htt]
\centering
\small
\caption{CP (\%) and AL for CIs of $c$ when there is no LLOD \label{CI.c.log}}
\begin{tabular}{cccccccccccc}
  \hline
 &$(n_0, n_1)$& \multicolumn{2}{c}{$(50, 50)$} & \multicolumn{2}{c}{$(100, 100)$} & \multicolumn{2}{c}{$(200, 200)$}& \multicolumn{2}{c}{$(50, 150)$} & \multicolumn{2}{c}{$(150, 50)$}\\
  \hline
   $J$& & CP & AL & CP & AL& CP & AL& CP & AL& CP & AL\\\hline
\multirow{5}{*}{$0.2$}& $\mathcal{I}_c$ & 93.9 & 5.06 & 95.3 & 3.44 & 95.0 & 2.44 & 94.3 & 4.14 & 95.3 & 3.82 \\ 
  &$\mathcal{I}_{cB}$ & 93.7 & 6.57 & 94.3 & 3.33 & 94.7 & 2.26 & 93.4 & 4.12 & 93.5 & 4.56 \\ 
  &$\mathcal{I}_{cG}$ & 93.4 & 6.95 & 94.6 & 4.22 & 95.1 & 2.89 & 93.1 & 4.67 & 95.8 & 4.42 \\ 
  &$\mathcal{I}_{cE}$ & 93.4 & 8.05 & 95.7 & 6.70 & 96.5 & 5.56 & 91.6 & 7.08 & 96.6 & 7.49 \\ 
  &$\mathcal{I}_{cK}$ & 95.5 & 8.95 & 95.0 & 6.48 & 94.8 & 4.67 & 95.6 & 7.53 & 95.2 & 7.21 \\ 
  &$\mathcal{I}_{cH}$ & 94.9 & 8.44 & 96.2 & 5.60 & 94.6 & 4.26 & 93.7 & 6.23 & 95.0 & 7.22 \\ \hline
\multirow{5}{*}{$0.4$}& $\mathcal{I}_c$ & 94.1 & 3.57 & 94.5 & 2.51 & 95.4 & 1.77 & 93.7 & 2.92 & 94.6 & 2.84 \\ 
 &$\mathcal{I}_{cB}$  & 93.8 & 3.36 & 92.9 & 2.37 & 94.6 & 1.68 & 93.7 & 2.84 & 92.8 & 2.52 \\ 
 &$\mathcal{I}_{cG}$ & 92.4 & 4.70 & 93.1 & 3.30 & 94.8 & 2.33 & 94.1 & 3.43 & 94.6 & 3.31 \\ 
  &$\mathcal{I}_{cE}$ & 94.7 & 6.48 & 96.8 & 5.39 & 97.0 & 4.45 & 94.0 & 5.78 & 96.7 & 6.05 \\ 
  &$\mathcal{I}_{cK}$& 94.8 & 5.74 & 93.5 & 4.00 & 93.2 & 3.07 & 94.6 & 4.89 & 92.1 & 4.16 \\ 
  &$\mathcal{I}_{cH}$ & 96.1 & 6.38 & 96.0 & 4.71 & 96.7 & 3.46 & 95.4 & 5.09 & 95.2 & 5.66 \\  \hline
  \multirow{5}{*}{$0.6$}& $\mathcal{I}_c$ & 92.9 & 3.35 & 94.1 & 2.38 & 95.2 & 1.67 & 93.9 & 2.69 & 94.7 & 2.73 \\ 
  &$\mathcal{I}_{cB}$  & 92.9 & 3.09 & 94.6 & 2.21 & 94.3 & 1.57 & 93.6 & 2.60 & 93.7 & 2.34 \\ 
  &$\mathcal{I}_{cG}$ & 91.6 & 4.71 & 92.9 & 3.37 & 93.9 & 2.39 & 92.2 & 3.33 & 92.3 & 3.28 \\ 
  &$\mathcal{I}_{cE}$ & 94.2 & 5.78 & 95.5 & 4.84 & 96.7 & 3.90 & 92.7 & 5.16 & 96.4 & 5.43 \\ 
 &$\mathcal{I}_{cK}$& 91.9 & 4.86 & 90.7 & 3.46 & 91.2 & 2.56 & 92.8 & 4.26 & 88.8 & 3.37 \\
 &$\mathcal{I}_{cH}$ & 96.9 & 6.02 & 96.1 & 4.46 & 96.6 & 3.18 & 94.7 & 4.72 & 98.0 & 5.38 \\  \hline
  \multirow{5}{*}{$0.8$}& $\mathcal{I}_c$ & 93.2 & 3.87 & 94.9 & 2.78 & 96.1 & 1.97 & 94.4 & 3.14 & 95.0 & 3.16 \\ 
  &$\mathcal{I}_{cB}$ & 94.0 & 3.26 & 94.9 & 2.31 & 94.7 & 1.64 & 93.8 & 2.79 & 95.3 & 2.49 \\ 
 &$\mathcal{I}_{cG}$& 90.2 & 5.81 & 92.8 & 4.34 & 93.0 & 3.05 & 92.8 & 4.15 & 89.1 & 3.72 \\ 
  &$\mathcal{I}_{cE}$ & 88.3 & 5.61 & 95.3 & 4.87 & 96.5 & 3.96 & 86.9 & 5.08 & 96.1 & 5.45 \\ 
  &$\mathcal{I}_{cK}$ & 92.8 & 5.12 & 91.6 & 3.59 & 90.2 & 2.68 & 92.3 & 4.39 & 89.4 & 3.43 \\ 
  &$\mathcal{I}_{cH}$& 97.6 & 5.75 & 96.2 & 4.48 & 95.3 & 3.31 & 92.4 & 4.44 & 98.6 & 5.40\\
  \hline
   \end{tabular}
\end{table}


\begin{table}[!htt]
\centering
\small
\caption{RB (\%) and MSE ($\times 100$) for point estimators of $J$ when the LLOD equals 15\% quantile of $F_0$}
\label{point.J.log.lod}
\begin{tabular}{cccccccccccc}
  \hline
 &$(n_0, n_1)$& \multicolumn{2}{c}{$(50, 50)$} & \multicolumn{2}{c}{$(100, 100)$} & \multicolumn{2}{c}{$(200, 200)$}& \multicolumn{2}{c}{$(50, 150)$} & \multicolumn{2}{c}{$(150, 50)$}\\ 
  \hline
  $J$& & RB & MSE & RB & MSE& RB & MSE& RB & MSE& RB & MSE\\ \hline
   \multirow{4}{*}{$0.2$}& $\hat{J}$ & 7.56 & 0.50 & 4.49 & 0.26 & 1.35 & 0.14 & 5.47 & 0.27 & 4.94 & 0.41 \\ 
   &$\hat{J}_{B}$ & 5.45 & 0.48 & 3.00 & 0.24 & 0.80 & 0.13 & 4.20 & 0.26 & 2.02 & 0.38 \\ 
  &$\hat{J}_{G}$  & 34.64 & 1.13 & 33.44 & 0.80 & 30.59 & 0.57 & 33.94 & 0.83 & 34.23 & 1.02 \\ 
  &$\hat{J}_{E}$ & 29.86 & 0.88 & 20.13 & 0.46 & 12.08 & 0.22 & 24.42 & 0.54 & 24.80 & 0.68 \\ \hline
  \multirow{4}{*}{$0.4$}& $\hat{J}$ & 3.76 & 0.56 & 2.21 & 0.27 & 0.62 & 0.15 & 2.71 & 0.30 & 1.92 & 0.45 \\ 
  &$\hat{J}_{B}$ & 3.29 & 0.52 & 1.79 & 0.25 & 0.51 & 0.14 & 2.26 & 0.28 & 1.26 & 0.40 \\ 
  &$\hat{J}_{G}$ & 8.26 & 0.69 & 8.27 & 0.39 & 6.94 & 0.25 & 7.99 & 0.42 & 8.23 & 0.59 \\ 
  &$\hat{J}_{E}$ & 14.11 & 0.91 & 9.35 & 0.44 & 5.61 & 0.23 & 11.48 & 0.54 & 10.90 & 0.67 \\ \hline
  \multirow{4}{*}{$0.6$}& $\hat{J}$ & 2.43 & 0.46 & 1.27 & 0.22 & 0.42 & 0.12 & 1.59 & 0.24 & 1.08 & 0.35 \\ 
  &$\hat{J}_{B}$ & 2.21 & 0.41 & 1.12 & 0.19 & 0.28 & 0.11 & 1.33 & 0.22 & 1.05 & 0.31 \\ 
  &$\hat{J}_{G}$ & 1.04 & 0.48 & 0.96 & 0.24 & 0.26 & 0.13 & 0.65 & 0.27 & 0.79 & 0.39 \\ 
  &$\hat{J}_{E}$ & 7.97 & 0.69 & 5.15 & 0.34 & 3.18 & 0.18 & 6.48 & 0.41 & 5.96 & 0.51 \\ \hline
  \multirow{4}{*}{$0.8$}& $\hat{J}$ & 1.73 & 0.26 & 0.84 & 0.13 & 0.33 & 0.07 & 1.03 & 0.14 & 0.82 & 0.20 \\ 
  &$\hat{J}_{B}$ & 1.27 & 0.21 & 0.63 & 0.10 & 0.17 & 0.06 & 0.73 & 0.12 & 0.68 & 0.16 \\ 
  &$\hat{J}_{G}$ & -0.64 & 0.29 & -0.62 & 0.15 & -0.92 & 0.08 & -1.00 & 0.17 & -0.71 & 0.23 \\ 
  &$\hat{J}_{E}$ & 4.60 & 0.39 & 2.91 & 0.19 & 1.87 & 0.10 & 3.70 & 0.23 & 3.24 & 0.28 \\ \hline
\end{tabular}
\end{table}

\begin{table}[!htt]
\centering
\small
\caption{RB (\%) and MSE ($\times100$) for point estimators of $c$ when the LLOD equals 15\% quantile of $F_0$ \label{point.c.log.lod}}
\begin{tabular}{cccccccccccc}
  \hline
 &$(n_0, n_1)$& \multicolumn{2}{c}{$(50, 50)$} & \multicolumn{2}{c}{$(100, 100)$} & \multicolumn{2}{c}{$(200, 200)$}& \multicolumn{2}{c}{$(50, 150)$} & \multicolumn{2}{c}{$(150, 50)$}\\
  \hline
  $J$& & RB & MSE & RB & MSE& RB & MSE& RB & MSE& RB & MSE\\ \hline
  \multirow{4}{*}{$0.2$}& $\hat{c}$  & -0.17 & 225.03 & -0.22 & 112.48 & 0.09 & 58.83 & -0.54 & 148.20 & 0.37 & 143.62 \\ 
  &$\hat{c}_{B}$ & -0.25 & 164.45 & -0.15 & 75.43 & 0.09 & 37.46 & -0.30 & 108.84 & 0.08 & 88.84 \\ 
  &$\hat{c}_{G}$ & -6.76 & 321.47 & -5.66 & 181.59 & -5.08 & 118.25 & -6.36 & 244.43 & -5.32 & 181.49 \\ 
  &$\hat{c}_{E}$ & -2.94 & 554.67 & -1.49 & 346.36 & -0.38 & 225.67 & -3.20 & 450.15 & -0.39 & 452.45 \\ \hline
  \multirow{4}{*}{$0.4$}& $\hat{c}$ & 0.07 & 116.35 & 0.06 & 61.77 & 0.06 & 30.91 & -0.25 & 81.46 & 0.51 & 80.19 \\ 
  &$\hat{c}_{B}$ & -0.23 & 92.05 & 0.06 & 45.31 & 0.02 & 22.06 & -0.16 & 60.59 & 0.15 & 54.55 \\ 
  &$\hat{c}_{G}$ & -1.49 & 148.7 & -0.54 & 72.78 & 0.02 & 36.59 & -0.84 & 82.54 & 0.04 & 70.48 \\ 
  &$\hat{c}_{E}$ & -1.59 & 346.94 & -0.53 & 244.02 & -0.37 & 150.03 & -1.53 & 285.84 & 0.21 & 300.37 \\ \hline
 \multirow{4}{*}{$0.6$}& $\hat{c}$ & -0.22 & 90.66 & -0.02 & 46.55 & -0.11 & 24.47 & -0.42 & 61.71 & 0.34 & 58.21 \\ 
  &$\hat{c}_{B}$ & -0.06 & 72.49 & 0.05 & 35.93 & -0.06 & 17.92 & -0.18 & 47.64 & 0.16 & 44.06 \\ 
  &$\hat{c}_{G}$ & -0.53 & 151.34 & 0.84 & 79.10 & 1.38 & 41.43 & 0.75 & 80.73 & 1.23 & 75.74 \\ 
  &$\hat{c}_{E}$ & -1.20 & 282.61 & -0.57 & 182.46 & -0.12 & 115.12 & -1.50 & 237.71 & 0.37 & 230.93 \\ \hline
  \multirow{4}{*}{$0.8$}& $\hat{c}$ & -0.41 & 105.32 & -0.23 & 51.60 & -0.26 & 27.89 & -0.68 & 68.33 & 0.09 & 67.62 \\ 
  &$\hat{c}_{B}$ & 0.00 & 74.60 & -0.01 & 36.17 & -0.12 & 17.85 & -0.29 & 53.10 & 0.25 & 41.60 \\ 
  &$\hat{c}_{G}$ & -1.29 & 235.86 & 0.54 & 116.04 & 1.16 & 65.54 & 1.00 & 114.92 & 0.40 & 120.53 \\ 
  &$\hat{c}_{E}$ & -1.37 & 292.03 & -0.55 & 189.9 & -0.77 & 128.77 & -1.46 & 232.99 & 0.44 & 235.25 \\ \hline
\end{tabular}
\end{table}

\begin{table}[!htt]
\centering
\small
\caption{CP (\%) and AL for CIs of $J$ when the LLOD equals 15\% quantile of $F_0$}
\label{CI.J.log.lod}
\begin{tabular}{cccccccccccc}
  \hline
 &$(n_0, n_1)$& \multicolumn{2}{c}{$(50, 50)$} & \multicolumn{2}{c}{$(100, 100)$} & \multicolumn{2}{c}{$(200, 200)$}& \multicolumn{2}{c}{$(50, 150)$} & \multicolumn{2}{c}{$(150, 50)$}\\
  \hline
   $J$& & CP & AL & CP & AL& CP & AL& CP & AL& CP & AL\\\hline
  \multirow{4}{*}{$0.2$}& $\mathcal{I}_J$ & 95.4 & 0.29 & 94.3 & 0.20 & 93.8 & 0.14 & 94.1 & 0.21 & 94.2 & 0.25 \\ 
  &$\mathcal{I}_{JB}$ & 95.5 & 0.28 & 94.5 & 0.20 & 93.7 & 0.14 & 95.3 & 0.20 & 93.8 & 0.25 \\ 
  &$\mathcal{I}_{JG}$ & 86.3 & 0.32 & 76.2 & 0.23 & 65.0 & 0.16 & 75.4 & 0.24 & 82.6 & 0.29 \\ 
  &$\mathcal{I}_{JE}$ & 75.4 & 0.26 & 77.7 & 0.19 & 79.0 & 0.14 & 70.2 & 0.20 & 77.8 & 0.24 \\ \hline
  \multirow{4}{*}{$0.4$}& $\mathcal{I}_J$ & 95.8 & 0.29 & 95.9 & 0.21 & 94.5 & 0.15 & 95.2 & 0.21 & 95.2 & 0.26 \\ 
  &$\mathcal{I}_{JB}$ & 95.8 & 0.29 & 95.4 & 0.21 & 93.6 & 0.14 & 95.6 & 0.21 & 95.0 & 0.26 \\ 
  &$\mathcal{I}_{JG}$ & 92.3 & 0.30 & 89.7 & 0.22 & 86.1 & 0.15 & 88.7 & 0.23 & 89.8 & 0.27 \\ 
  &$\mathcal{I}_{JE}$ & 80.0 & 0.27 & 82.5 & 0.20 & 83.6 & 0.15 & 75.2 & 0.21 & 82.8 & 0.25 \\ \hline
  \multirow{4}{*}{$0.6$}& $\mathcal{I}_J$ & 95.6 & 0.26 & 96.0 & 0.19 & 94.7 & 0.13 & 94.9 & 0.19 & 95.3 & 0.23 \\ 
   &$\mathcal{I}_{JB}$& 95.3 & 0.25 & 95.9 & 0.18 & 94.5 & 0.13 & 95.6 & 0.19 & 96.0 & 0.22 \\ 
  &$\mathcal{I}_{JG}$  & 93.4 & 0.27 & 94.5 & 0.19 & 93.6 & 0.14 & 93.7 & 0.20 & 92.5 & 0.24 \\ 
  &$\mathcal{I}_{JE}$ & 81.4 & 0.25 & 81.7 & 0.18 & 83.5 & 0.14 & 75.9 & 0.19 & 82.3 & 0.23 \\ \hline
  \multirow{4}{*}{$0.8$}& $\mathcal{I}_J$ & 96.8 & 0.20 & 95.6 & 0.14 & 95 & 0.10 & 95.5 & 0.15 & 95.1 & 0.18 \\ 
   &$\mathcal{I}_{JB}$ & 94.9 & 0.18 & 96.2 & 0.13 & 95.2 & 0.09 & 96.1 & 0.14 & 95.1 & 0.16 \\ 
  &$\mathcal{I}_{JG}$  & 94.2 & 0.21 & 95.0 & 0.15 & 95.2 & 0.11 & 95.1 & 0.15 & 93.7 & 0.19 \\ 
  &$\mathcal{I}_{JE}$ & 81.5 & 0.18 & 85.2 & 0.14 & 85.8 & 0.10 & 77.6 & 0.14 & 83.7 & 0.17 \\\hline
   \end{tabular}
\end{table}

\begin{table}[!htt]
\centering
\small
\caption{CP (\%) and AL for CIs of $c$ when the LLOD equals 15\% quantile of $F_0$}\label{CI.c.log.lod}
\begin{tabular}{cccccccccccc}
  \hline
 &$(n_0, n_1)$& \multicolumn{2}{c}{$(50, 50)$} & \multicolumn{2}{c}{$(100, 100)$} & \multicolumn{2}{c}{$(200, 200)$}& \multicolumn{2}{c}{$(50, 150)$} & \multicolumn{2}{c}{$(150, 50)$}\\
  \hline
   $J$& & CP & AL & CP & AL& CP & AL& CP & AL& CP & AL\\\hline
   \multirow{4}{*}{$0.2$}& $\mathcal{I}_c$ & 93.7 & 5.96 & 94.3 & 4.22 & 95.0 & 3.05 & 93.5 & 4.84 & 93.5 & 4.83 \\ 
  &$\mathcal{I}_{cB}$ & 84.6 & 5.35 & 84.1 & 2.50 & 83.5 & 1.70 & 84.9 & 3.20 & 79.1 & 2.95 \\ 
  &$\mathcal{I}_{cG}$ & 82.8 & 5.55 & 81.6 & 3.73 & 71.5 & 2.57 & 81.9 & 4.41 & 82.2 & 4.22 \\ 
  &$\mathcal{I}_{cE}$ & 92.9 & 7.55 & 95.3 & 6.47 & 95.3 & 5.42 & 90.8 & 6.77 & 95.4 & 7.16 \\ \hline
  \multirow{4}{*}{$0.4$}& $\mathcal{I}_c$ & 93.0 & 4.27 & 93.4 & 3.11 & 95.3 & 2.27 & 92.6 & 3.52 & 91.9 & 3.55 \\ 
  &$\mathcal{I}_{cB}$ & 88.8 & 3.03 & 88.9 & 2.12 & 88.7 & 1.51 & 87.1 & 2.35 & 89.2 & 2.41 \\ 
  &$\mathcal{I}_{cG}$ & 94.4 & 4.58 & 93.2 & 3.29 & 94.3 & 2.34 & 93.6 & 3.51 & 94.2 & 3.22 \\ 
  &$\mathcal{I}_{cE}$ & 94.5 & 6.31 & 96.2 & 5.25 & 96.8 & 4.38 & 93.1 & 5.64 & 96.5 & 5.96 \\ \hline
 \multirow{4}{*}{$ 0.6$}& $\mathcal{I}_c$ & 91.7 & 3.72 & 93.2 & 2.72 & 94.5 & 1.97 & 93.5 & 3.04 & 94.3 & 3.12 \\ 
  &$\mathcal{I}_{cB}$ & 91.6 & 3.04 & 92.2 & 2.15 & 92.0 & 1.52 & 90.9 & 2.36 & 93.4 & 2.44 \\ 
  &$\mathcal{I}_{cG}$ & 94.5 & 4.70 & 94.0 & 3.36 & 93.7 & 2.42 & 94.3 & 3.40 & 94.6 & 3.26 \\ 
  &$\mathcal{I}_{cE}$ & 93.9 & 5.74 & 95.3 & 4.79 & 96.5 & 3.88 & 92.3 & 5.10 & 95.8 & 5.40 \\ \hline
  \multirow{4}{*}{$0.8$}& $\mathcal{I}_c$ & 93.6 & 3.97 & 94.8 & 2.91 & 95.0 & 2.10 & 93.5 & 3.23 & 94.1 & 3.33 \\ 
  &$\mathcal{I}_{cB}$ & 94.2 & 3.34 & 95.8 & 2.36 & 95.6 & 1.67 & 94.3 & 2.77 & 94.4 & 2.60 \\ 
  &$\mathcal{I}_{cG}$ & 94.8 & 5.86 & 94.9 & 4.29 & 94.7 & 3.03 & 95.6 & 4.21 & 94.7 & 3.84 \\ 
  &$\mathcal{I}_{cE}$ & 88.6 & 5.57 & 95.1 & 4.84 & 96.0 & 3.94 & 86.5 & 5.06 & 95.5 & 5.48 \\ \hline
   \end{tabular}
\end{table}

\begin{table}[!htt]
\centering
\small
\caption{RB (\%) and MSE ($\times 100$) for point estimators of $J$ when the LLOD equals 30\% quantile of $F_0$}
\label{point.J.log.lod3}
\begin{tabular}{cccccccccccc}
  \hline
 &$(n_0, n_1)$& \multicolumn{2}{c}{$(50, 50)$} & \multicolumn{2}{c}{$(100, 100)$} & \multicolumn{2}{c}{$(200, 200)$}& \multicolumn{2}{c}{$(50, 150)$} & \multicolumn{2}{c}{$(150, 50)$}\\
  \hline
  $J$& & RB & MSE & RB & MSE& RB & MSE& RB & MSE& RB & MSE\\ \hline
  \multirow{4}{*}{$0.2$}& $\hat{J}$ & 8.39 & 0.52 & 4.76 & 0.26 & 1.57 & 0.14 & 6.10 & 0.28 & 5.56 & 0.43 \\ 
  &$\hat{J}_{B}$ & 4.89 & 0.50 & 2.83 & 0.25 & 0.80 & 0.14 & 4.59 & 0.26 & -1.52 & 0.44 \\ 
  &$\hat{J}_{G}$ & 47.93 & 1.70 & 45.32 & 1.25 & 42.55 & 0.94 & 46.74 & 1.32 & 46.29 & 1.50 \\ 
  &$\hat{J}_{E}$ & 29.86 & 0.89 & 20.03 & 0.46 & 12.13 & 0.22 & 24.31 & 0.53 & 24.86 & 0.69 \\ \hline
  \multirow{4}{*}{$0.4$}& $\hat{J}$ & 3.81 & 0.57 & 2.26 & 0.27 & 0.59 & 0.16 & 2.83 & 0.31 & 1.99 & 0.46 \\ 
  &$\hat{J}_{B}$ & 3.36 & 0.52 & 1.95 & 0.25 & 0.45 & 0.14 & 2.35 & 0.28 & 1.45 & 0.43 \\ 
  &$\hat{J}_{G}$ & 10.94 & 0.85 & 10.33 & 0.49 & 9.08 & 0.32 & 10.50 & 0.55 & 10.25 & 0.68 \\ 
  &$\hat{J}_{E}$ & 14.00 & 0.90 & 9.36 & 0.43 & 5.54 & 0.23 & 11.46 & 0.54 & 10.91 & 0.67 \\ \hline
  \multirow{4}{*}{$0.6$}& $\hat{J}$ & 2.52 & 0.48 & 1.28 & 0.22 & 0.42 & 0.12 & 1.73 & 0.26 & 1.06 & 0.37 \\ 
  &$\hat{J}_{B}$ & 2.21 & 0.41 & 1.10 & 0.19 & 0.27 & 0.11 & 1.36 & 0.22 & 1.05 & 0.31 \\ 
  &$\hat{J}_{G}$ & 0.65 & 0.54 & 0.41 & 0.26 & -0.31 & 0.14 & 0.30 & 0.32 & 0.18 & 0.42 \\ 
  &$\hat{J}_{E}$ & 7.95 & 0.69 & 5.15 & 0.34 & 3.16 & 0.17 & 6.51 & 0.41 & 5.95 & 0.51 \\ \hline
  \multirow{4}{*}{$0.8$}& $\hat{J}$ & 1.85 & 0.27 & 0.93 & 0.13 & 0.35 & 0.07 & 1.14 & 0.14 & 0.77 & 0.20 \\ 
  &$\hat{J}_{B}$ & 1.16 & 0.21 & 0.56 & 0.10 & 0.15 & 0.06 & 0.62 & 0.12 & 0.52 & 0.16 \\ 
  &$\hat{J}_{G}$ & -1.38 & 0.36 & -1.51 & 0.18 & -1.92 & 0.11 & -1.78 & 0.22 & -1.84 & 0.26 \\ 
  &$\hat{J}_{E}$ & 4.57 & 0.38 & 2.91 & 0.19 & 1.86 & 0.10 & 3.65 & 0.23 & 3.17 & 0.28 \\  \hline
\end{tabular}
\end{table}

\begin{table}[!htt]
\centering
\small
\caption{RB (\%) and MSE ($\times100$) for point estimators of $c$ when the LLOD equals 30\% quantile of $F_0$}
\label{point.c.log.lod3}
\begin{tabular}{cccccccccccc}
  \hline
 &$(n_0, n_1)$& \multicolumn{2}{c}{$(50, 50)$} & \multicolumn{2}{c}{$(100, 100)$} & \multicolumn{2}{c}{$(200, 200)$}& \multicolumn{2}{c}{$(50, 150)$} & \multicolumn{2}{c}{$(150, 50)$}\\ 
  \hline
  $J$& & RB & MSE & RB & MSE& RB & MSE& RB & MSE& RB & MSE\\ \hline\multirow{4}{*}{$0.2$}& $\hat{c}$ & 0.65 & 258.41 & 0.43 & 136.56 & 0.48 & 73.65 & 0.18 & 174.52 & 1.09 & 184.57 \\ 
  &$\hat{c}_{B}$ & -0.50 & 281.55 & -0.31 & 93.98 & 0.06 & 43.73 & -0.48 & 127.55 & -1.24 & 154.60 \\ 
  &$\hat{c}_{G}$ & -3.91 & 233.83 & -2.66 & 117.81 & -2.33 & 62.51 & -3.21 & 154.52 & -2.20 & 121.42 \\ 
  &$\hat{c}_{E}$ & -2.88 & 543.75 & -1.48 & 340.90 & -0.36 & 227.91 & -3.10 & 447.73 & -0.34 & 448.97 \\ \hline
 \multirow{4}{*}{$0.4$}& $\hat{c}$ & 0.91 & 138.59 & 0.60 & 75.44 & 0.29 & 39.32 & 0.36 & 96.11 & 1.06 & 96.88 \\ 
  &$\hat{c}_{B}$ & 0.09 & 101.96 & 0.15 & 52.63 & 0.08 & 24.67 & -0.20 & 69.73 & 0.22 & 69.45 \\ 
  &$\hat{c}_{G}$ & 1.86 & 155.21 & 2.88 & 95.61 & 3.35 & 69.84 & 2.76 & 103.54 & 3.45 & 103.53 \\ 
  &$\hat{c}_{E}$ & -1.56 & 349.31 & -0.54 & 243.76 & -0.41 & 151.41 & -1.56 & 283.63 & 0.22 & 300.01 \\ \hline
 \multirow{4}{*}{$ 0.6$}& $\hat{c}$ & 0.38 & 103.45 & 0.34 & 51.80 & 0.14 & 28.56 & 0.00 & 69.41 & 0.80 & 68.51 \\ 
  &$\hat{c}_{B}$ & 0.13 & 79.98 & 0.14 & 38.93 & -0.01 & 19.06 & -0.21 & 51.75 & 0.47 & 49.46 \\ 
  &$\hat{c}_{G}$ & 2.69 & 171.73 & 4.13 & 122.81 & 4.62 & 100.98 & 4.08 & 130.54 & 4.46 & 135.43 \\ 
  &$\hat{c}_{E}$ & -1.12 & 281.90 & -0.61 & 183.18 & -0.09 & 115.00 & -1.54 & 239.01 & 0.35 & 228.98 \\ \hline
  \multirow{4}{*}{$0.8$}& $\hat{c}$ & 0.00 & 111.15 & 0.02 & 56.67 & -0.10 & 29.70 & -0.44 & 74.27 & 0.41 & 75.64 \\ 
  &$\hat{c}_{B}$ & 0.13 & 77.44 & 0.06 & 37.09 & -0.07 & 18.56 & -0.21 & 52.89 & 0.39 & 44.46 \\ 
  &$\hat{c}_{G}$ & 1.28 & 236.85 & 3.12 & 156.06 & 3.86 & 113.08 & 3.47 & 165.46 & 2.88 & 155.10 \\ 
  &$\hat{c}_{E}$ & -1.36 & 292.78 & -0.56 & 189.16 & -0.74 & 128.34 & -1.43 & 232.44 & 0.44 & 233.30 \\ \hline
\end{tabular}
\end{table}

\begin{table}[!htt]
\centering
\small
\caption{CP (\%) and AL for CIs of $J$ when the LLOD equals 30\% quantile of $F_0$}
\label{CI.J.log.lod3}
\begin{tabular}{cccccccccccc}
  \hline
 &$(n_0, n_1)$& \multicolumn{2}{c}{$(50, 50)$} & \multicolumn{2}{c}{$(100, 100)$} & \multicolumn{2}{c}{$(200, 200)$}& \multicolumn{2}{c}{$(50, 150)$} & \multicolumn{2}{c}{$(150, 50)$}\\
  \hline
   $J$& & CP & AL & CP & AL& CP & AL& CP & AL& CP & AL\\\hline
   \multirow{4}{*}{$0.2$}& $\mathcal{I}_J$ & 94.7 & 0.29 & 94.2 & 0.20 & 94.1 & 0.14 & 93.8 & 0.21 & 94.5 & 0.25 \\ 
  &$\mathcal{I}_{JB}$ & 94.5 & 0.28 & 93.5 & 0.20 & 93.5 & 0.14 & 95.4 & 0.21 & 90.5 & 0.23 \\ 
  &$\mathcal{I}_{JG}$ & 78.3 & 0.35 & 67.7 & 0.25 & 48.2 & 0.18 & 64.5 & 0.27 & 74.7 & 0.31 \\ 
  &$\mathcal{I}_{JE}$ & 71.2 & 0.24 & 75.1 & 0.18 & 76.0 & 0.13 & 67.2 & 0.19 & 72.8 & 0.22 \\ \hline
  \multirow{4}{*}{$0.4$}& $\mathcal{I}_J$ & 96.5 & 0.29 & 96.0 & 0.21 & 94.4 & 0.15 & 95.3 & 0.21 & 95.2 & 0.26 \\ 
  &$\mathcal{I}_{JB}$ & 96.0 & 0.29 & 95.2 & 0.20 & 93.9 & 0.14 & 95.5 & 0.21 & 94.1 & 0.25 \\ 
  &$\mathcal{I}_{JG}$ & 89.2 & 0.32 & 87.8 & 0.23 & 83.5 & 0.16 & 86.4 & 0.24 & 88.2 & 0.28 \\ 
  &$\mathcal{I}_{JE}$ & 75.3 & 0.25 & 78.9 & 0.19 & 80.7 & 0.14 & 72.5 & 0.20 & 78.3 & 0.23 \\ \hline
 \multirow{4}{*}{$0.6$}& $\mathcal{I}_J$ & 95.2 & 0.26 & 96.0 & 0.19 & 94.7 & 0.13 & 94.8 & 0.19 & 94.7 & 0.23 \\ 
 &$\mathcal{I}_{JB}$ & 95.4 & 0.25 & 96.0 & 0.18 & 94.8 & 0.13 & 95.5 & 0.19 & 95.6 & 0.23 \\ 
 &$\mathcal{I}_{JG}$ & 94.0 & 0.29 & 94.3 & 0.20 & 93.8 & 0.14 & 94.1 & 0.22 & 92.3 & 0.25 \\ 
 &$\mathcal{I}_{JE}$ & 78.4 & 0.24 & 80.3 & 0.18 & 82.1 & 0.13 & 75.2 & 0.18 & 80.4 & 0.21 \\ \hline
  \multirow{4}{*}{$0.8$}& $\mathcal{I}_J$ & 96.4 & 0.20 & 95.2 & 0.14 & 95.2 & 0.10 & 95.2 & 0.15 & 95.4 & 0.18 \\ 
  &$\mathcal{I}_{JB}$ & 95.4 & 0.19 & 96.1 & 0.14 & 94.5 & 0.10 & 96.6 & 0.14 & 94.4 & 0.17 \\ 
  &$\mathcal{I}_{JG}$ & 94.6 & 0.23 & 95.2 & 0.16 & 93.5 & 0.11 & 94.9 & 0.17 & 95.0 & 0.20 \\ 
  &$\mathcal{I}_{JE}$ & 80.3 & 0.18 & 83.9 & 0.14 & 84.4 & 0.10 & 76.6 & 0.14 & 82.2 & 0.16 \\ \hline
      \end{tabular}
\end{table}

\begin{table}[!htt]
\centering
\small
\caption{CP (\%) and AL for CIs of $c$ when the LLOD equals 30\% quantile of $F_0$}
\label{CI.c.log.lod3}
\begin{tabular}{cccccccccccc}
  \hline
 &$(n_0, n_1)$& \multicolumn{2}{c}{$(50, 50)$} & \multicolumn{2}{c}{$(100, 100)$} & \multicolumn{2}{c}{$(200, 200)$}& \multicolumn{2}{c}{$(50, 150)$} & \multicolumn{2}{c}{$(150, 50)$}\\
  \hline
   $J$& & CP & AL & CP & AL& CP & AL& CP & AL& CP & AL\\\hline
  \multirow{4}{*}{$0.2$}& $\mathcal{I}_c$  & 90.0 & 6.50 & 91.2 & 4.63 & 93.7 & 3.41 & 92.0 & 5.26 & 89.9 & 5.26 \\ 
  &$\mathcal{I}_{cB}$ & 80.9 & 5.72 & 80.6 & 2.76 & 79.3 & 1.73 & 75.9 & 2.88 & 80.6 & 4.26 \\ 
  &$\mathcal{I}_{cG}$ & 90.2 & 5.64 & 91.6 & 3.84 & 90.5 & 2.69 & 91.2 & 4.38 & 91.3 & 4.15 \\ 
  &$\mathcal{I}_{cE}$ & 92.4 & 7.15 & 95.0 & 6.28 & 95.0 & 5.27 & 90.3 & 6.44 & 94.4 & 6.88 \\ \hline
 \multirow{4}{*}{$0.4$}& $\mathcal{I}_c$ & 91.1 & 4.67 & 93.2 & 3.40 & 93.1 & 2.51 & 91.9 & 3.89 & 91.4 & 3.83 \\ 
  &$\mathcal{I}_{cB}$ & 86.6 & 3.16 & 86.5 & 2.12 & 86.0 & 1.49 & 80.2 & 2.12 & 90.3 & 3.02 \\ 
  &$\mathcal{I}_{cG}$ & 95.1 & 4.72 & 92.5 & 3.36 & 85.7 & 2.39 & 92.6 & 3.62 & 90.1 & 3.28 \\ 
  &$\mathcal{I}_{cE}$ & 94.0 & 6.15 & 95.8 & 5.17 & 96.2 & 4.34 & 92.7 & 5.48 & 96.8 & 5.85 \\ \hline
 \multirow{4}{*}{$ 0.6$}& $\mathcal{I}_c$ & 92.2 & 3.97 & 94.0 & 2.94 & 94.0 & 2.14 & 93.7 & 3.30 & 93.7 & 3.33 \\ 
  &$\mathcal{I}_{cB}$ & 89.6 & 2.98 & 91.0 & 2.08 & 89.6 & 1.47 & 85.3 & 2.14 & 92.0 & 2.57 \\ 
  &$\mathcal{I}_{cG}$ & 95.0 & 4.77 & 89.3 & 3.39 & 78.0 & 2.43 & 89.6 & 3.55 & 87.2 & 3.29 \\ 
  &$\mathcal{I}_{cE}$ & 93.6 & 5.61 & 95.5 & 4.75 & 96.4 & 3.85 & 92.1 & 5.00 & 96.0 & 5.35 \\ \hline
  \multirow{4}{*}{$0.8$}& $\mathcal{I}_c$ & 92.7 & 4.08 & 94.1 & 3.02 & 95.0 & 2.20 & 92.5 & 3.34 & 93.3 & 3.45 \\ 
  &$\mathcal{I}_{cB}$ & 93.1 & 3.30 & 95.2 & 2.33 & 94.7 & 1.64 & 93.2 & 2.63 & 95.0 & 2.66 \\ 
  &$\mathcal{I}_{cG}$ & 96.3 & 5.97 & 94.2 & 4.25 & 87.6 & 3.04 & 93.4 & 4.34 & 94.2 & 3.91 \\ 
  &$\mathcal{I}_{cE}$ & 88.3 & 5.55 & 94.6 & 4.81 & 96.1 & 3.95 & 86.2 & 5.00 & 95.4 & 5.42 \\ \hline
   \end{tabular}
\end{table}

\end{document}